\documentclass{article}

\usepackage{microtype}
\usepackage{graphicx}
\usepackage{booktabs} 

\usepackage{hyperref}

\usepackage[preprint]{icml2026}

\usepackage{amsmath}
\usepackage{amssymb}
\usepackage{mathtools}
\usepackage{amsthm}

\usepackage{times}
\usepackage{subfigure}
\usepackage{multirow}
\usepackage{threeparttable}
\usepackage{tikz}
\usepackage{colortbl}
\usepackage[most]{tcolorbox}
\usepackage{listings}
\usepackage{inconsolata}
\usepackage{tikz}
\usepackage[textsize=tiny]{todonotes}
\usepackage{xspace}
\usepackage{wrapfig,lipsum,booktabs}

\usepackage{caption}

\usepackage{hyperref}
\usepackage{courier}
\usepackage{graphicx}

\usepackage[capitalize,noabbrev]{cleveref}

\theoremstyle{plain}
\newtheorem{theorem}{Theorem}[section]
\newtheorem{proposition}[theorem]{Proposition}
\newtheorem{lemma}[theorem]{Lemma}

\theoremstyle{definition}
\newtheorem{definition}[theorem]{Definition}

\theoremstyle{remark}
\newtheorem{remark}[theorem]{Remark}

\newcommand{\E}{\mathbb{E}}
\newcommand{\Prob}{\mathbb{P}}

\newcommand{\hide}[1]{}
\renewcommand{\algorithmiccomment}[1]{\hfill $\triangleright$\ #1}

\newcommand{\name}{\textsc{RisCoSet}\xspace}

\usepackage[textsize=tiny]{todonotes}

\icmltitlerunning{Uncertainty Quantification for LLM-based Code Generation}

\begin{document}

\twocolumn[
  \icmltitle{Uncertainty Quantification for LLM-based Code Generation}



  \icmlsetsymbol{equal}{*}




  \begin{icmlauthorlist}
    \icmlauthor{Senrong Xu}{yyy}
    \icmlauthor{Yuhao Tan}{yyy}
    \icmlauthor{Yanke Zhou}{yyy}
    \icmlauthor{Guangyuan Wu}{yyy}
    \icmlauthor{Zenan Li}{yyyy}\\
    \icmlauthor{Yuan Yao}{yyy}
    \icmlauthor{Taolue Chen}{yyyyy}
    \icmlauthor{Feng Xu}{yyy}
    \icmlauthor{Xiaoxing Ma}{yyy}
  \end{icmlauthorlist}

  \icmlaffiliation{yyy}{State Key Lab of Novel Software Technology, Nanjing University}
  \icmlaffiliation{yyyy}{ETH Zürich}
  \icmlaffiliation{yyyyy}{Birkbeck, University of London}

  \icmlcorrespondingauthor{Senrong Xu}{srxu@smail.nju.edu.cn}
  \icmlcorrespondingauthor{Yuan Yao}{y.yao@nju.edu.cn}

  \icmlkeywords{Machine Learning, ICML}

  \vskip 0.3in
]



\printAffiliationsAndNotice{}  

\begin{abstract}

Prediction sets provide a theoretically grounded framework for quantifying uncertainty in machine learning models. Adapting them to structured generation tasks, in particular, large language model (LLM) based code generation, remains a challenging problem. An existing attempt proposes PAC prediction sets but is limited by its strong monotonicity assumption on risk and single-label classification framework, 
which severely limits the space of candidate programs and cannot accommodate the multiple valid outputs inherent to code generation.
To address these limitations, 
we propose an approach \name that leverages multiple hypothesis testing to construct risk-controlling predictions for LLM-based code generation. Given a trained code generation model, we produce a prediction set represented by a partial program, which is guaranteed to contain a correct solution with high confidence. 
Extensive experiments on three LLMs demonstrate the effectiveness of the proposed method. For instance, compared with the state-of-the-art, our method can significantly reduce the code removal by up to 24.5\%, at the same level of risk.

\end{abstract}


\section{Introduction} \label{sec:intro}
\begin{figure*}[t]
 \centering
 \subfigure{
 \includegraphics[width=3.3in]{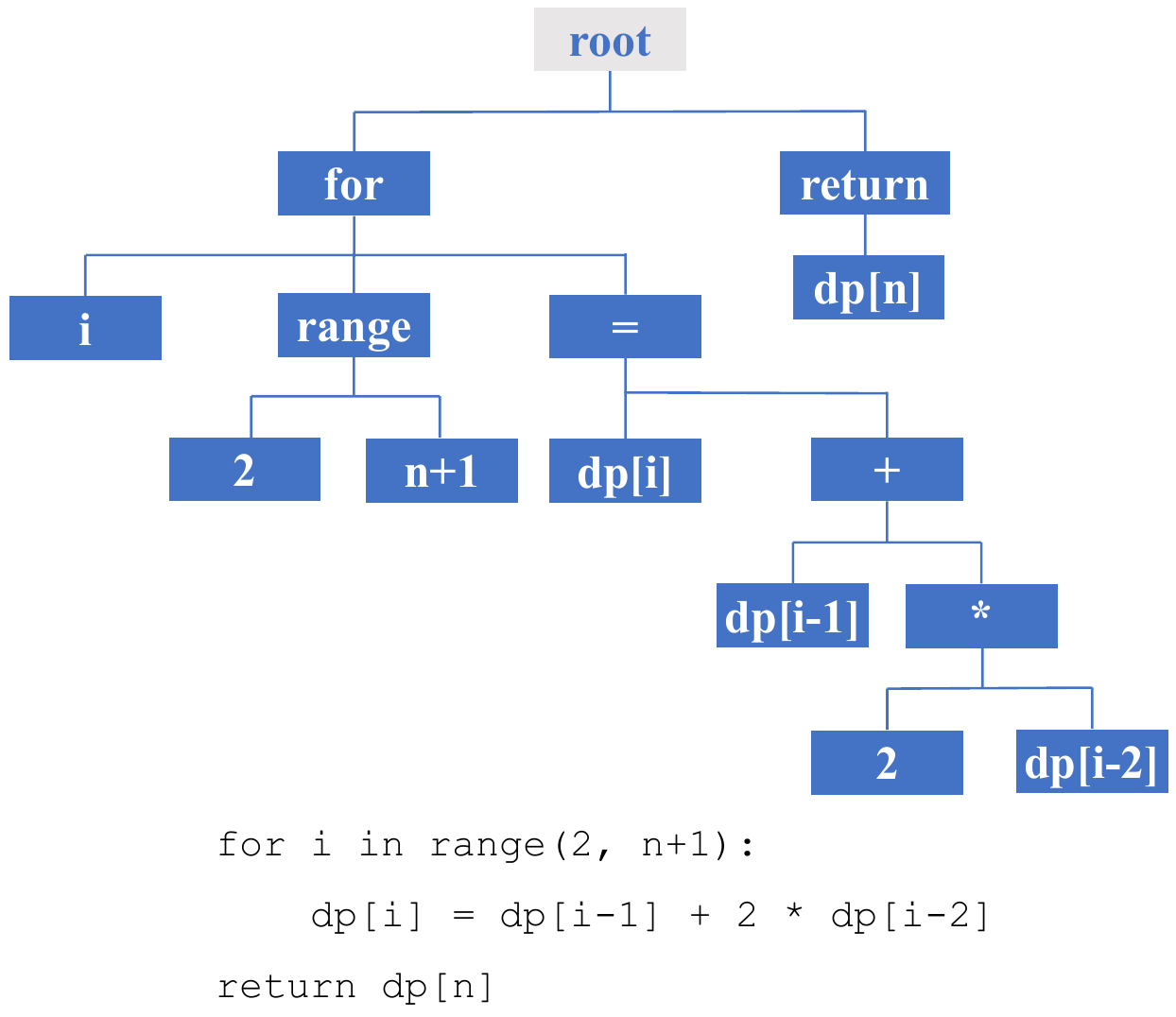}}
 \subfigure{
\includegraphics[width=3.1in]{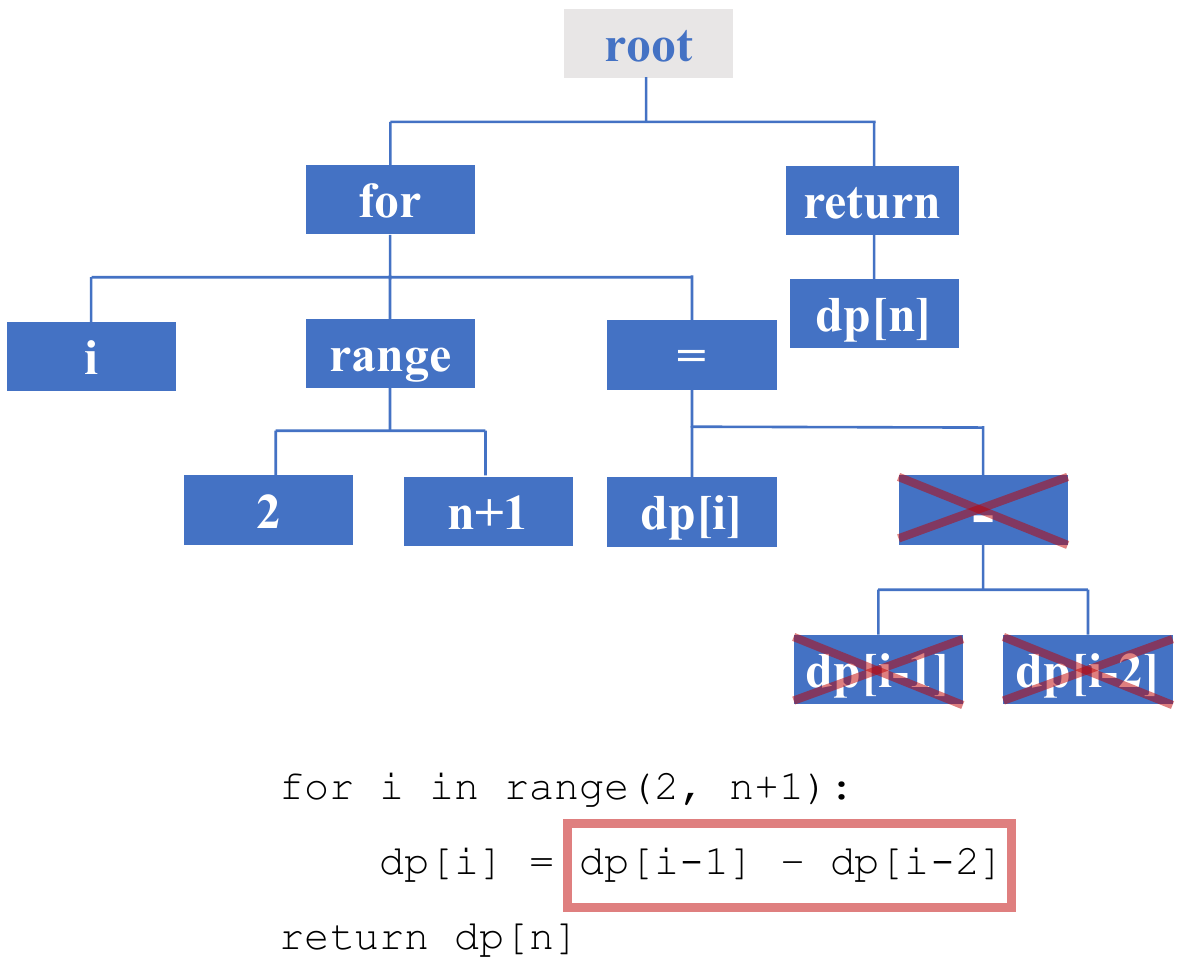}}
 \caption{An illustrative example from MBPP. The left part is a correct code snippet, and the right part is a generated but incorrect one.
 \name removes three nodes in the AST, resulting a prediction set (i.e., a partial program) that contains  the correct program.}
 \label{F:example}
 \vspace{-0.5em}
\end{figure*}

Large language models (LLMs) 
have been widely adopted by millions of users worldwide due to their strong performance~\cite{achiam2023gpt,team2023gemini}. 
However, they often hallucinate facts and may generate toxic or biased outputs~\cite{nadeau2024benchmarking}, limiting 
their use in high-stakes applications such as code generation and legal analysis~\cite{liu2024exploring,dahl2024large}. 

As an uncertainty quantification method, conformal prediction (CP) can provide model-agnostic, distribution-free uncertainty guarantees for machine learning models, and thus has attracted recent attention~\cite{vovk2005algorithmic,angelopoulos2022conformal}. 
For example, several studies have applied CP into LLM generation by identifying sound components within the LLM output~\cite{quach2024conformallanguagemodeling,cherian2024largelanguagemodelvalidity}.
However, these methods assume near independence among components, an assumption that holds for sentences in NLP but fails for structured outputs like code, where components interact systemically~\cite{casalnuovo2019studying,chen2022codet}. 

To address this, \citet{khakhar2023pac} proposed {PAC (probably approximately correct) prediction sets} for code generation, constructing structured prediction sets and providing statistical guarantees via partial programs~\cite{guo2021learning}. 
However, this approach suffers from two key limitations: (1) it relies on a strong monotonicity constraint that restricts the space of admissible partial programs, which may lead to a suboptimal solution; (2) it is designed for single-label classification and does not easily accommodate the multiple valid programs (e.g., multiple syntactically different programs with identical, correct semantics) that naturally arise in code generation.



In this paper, we introduce a new method, \name, for constructing risk-controlling prediction sets tailored to code generation.
To overcome the limitations in PAC prediction sets, we reframe code uncertainty quantification as a multi-label classification problem, and adopt the Learn Then Test (LTT) framework~\cite{angelopoulos2025learn}, a recent risk-controlling method, to construct structured prediction sets with rigorous statistical guarantees. 
That is, given an LLM-generated code, 
\name yields 
a partial program (i.e., a subtree of the program's abstract syntax tree (AST), see Fig.~\ref{F:example}), which is guaranteed, with high confidence, to contain a correct solution. 

\hide{
The flexibility and scalability of LTT framework makes it suitable to diverse tasks, such as classification, regression prediction, and out-of-distribution detection,
However, the search space over prediction sets of tasks mentioned above is simple and only has one dimension. 
This very different on structured data, whose search space are complex and high dimensional.
Therefore, we design a special risk function for such structured data in the LTT framework.
In addition, to collect more valid programs as true labels in multi-label settings, we need to sample more candidate programs from LLMs and execute their test cases to verify correctness during calibration.  
In order to lower such computation costs, we propose a selective execution strategy that efficiently verifies LLM-sampled programs to augment the true label set. 
This strategy decides whether a sample program needs to be executed against test cases with a well-calibrated threshold.
}

However, applying LTT into code generation is still challenging. The first challenge lies in the complexity of code search space.
Unlike traditional applications such as classification and regression whose search space is typically one-dimensional, 
the search space of structured data (such as code) is complex and multi-dimensional~\cite{solar2008program}.
To this end, we design a tailored risk function and 
the corresponding optimization problem for obtaining partial programs. 
Furthermore, in multi-label settings, we need to sample additional candidate programs from LLMs and verify their correctness via test-case execution during calibration.
This is a computationally expensive step. 
To reduce this cost, we propose a selective execution strategy that efficiently verifies the correctness of sampled programs by using a well-calibrated threshold to decide which programs require test execution.

\noindent\textbf{Example.}
Fig.~\ref{F:example} illustrates a real example 
from the MBPP dataset~\cite{austin2021program}. 
The left code snippet is from a correct program, while the right one is generated by the model and is incorrect. 
Our method, \name, removes 
the uncertain nodes in the AST (highlighted in red), resulting in a partial program (AST subtree) which can be completed to match a correct program.
This example also shows the potential implication of our method: \name can provide an almost-correct partial program of a given generated program, helping to mitigate code hallucinations~\cite{liu2024exploring} and reduce the cost of manual review~\cite{zhong2017boosting}. 
%


\noindent\textbf{Contributions.} 
The main contributions of the paper are summarized as:
(1) we specialize an LTT-based code generation approach \name for uncertainty quantification of LLM-based code generation, which reframes the construction of prediction sets for code as a multi-label problem with weaker constraints; 
(2) we propose a selective execution strategy to balance the trade-off between risk-controlling and sampling costs; 
(3) We conduct comprehensive experiments across three LLMs and three code generation datasets, confirming the effectiveness of our method.
To the best of our knowledge, this is the first work to introduce the LTT framework into structured code generation tasks. 


\noindent\textbf{Structure.} 
Section~\ref{sec:pre} briefly introduces the background. 
Section~\ref{sec:vs} compares PAC and LTT in code generation. 
Section~\ref{sec:method} presents the proposed methods and Section~\ref{sec:exp} reports the experimental results. Section~\ref{sec:rel} discusses the related work, and Section~\ref{sec:con} concludes the paper.


\section{Background} \label{sec:pre} 

\textbf{Notations}.
For any natural number $n$, let $[n]:=\{1,\dots,n\}$.
Let $\mathcal{D}$ denote a distribution over $\mathcal{X}\times\mathcal{Y}$, and let
$\mathcal{D}_\mathrm{cal}=\{(X_i,Y_i)\}_{i=1}^n \stackrel{\mathrm{i.i.d.}}{\sim} \mathcal{D}$
denote a calibration set.
In the single-label case, $Y_i\in\mathcal{Y}$; in the multi-label case, $Y_i\subseteq\mathcal{Y}$.
Let $(X_{n+1},Y_{n+1})\sim \mathcal{D}$ be a test sample.
An abstract syntax tree (AST) of a program is denoted by
$T=(\mathcal{V},\mathcal{E})$, where $\mathcal{V}$ and $\mathcal{E}$ are the sets of nodes and edges, respectively.

\subsection{PAC Prediction Set}
The objective of a PAC prediction~\cite{park2020pac} is to produce a prediction set, instead of a single prediction, that satisfies statistical guarantees while remaining small in size.
Formally, for a test input $X_{n+1}\in \mathcal{X}$, PAC prediction set $C_\tau(X_{n+1})\subseteq \mathcal{Y}$ ensures that the true label $Y_{n+1}$ is contained within the set for at least a $1-\varepsilon$ fraction of the population, with a probability at least $1-\delta$, i.e.,
\begin{equation*}\label{eq:PAC}
\Prob\Big(\Prob_{(X,Y)\sim \mathcal{D}}\big(Y\in C_\tau(X)\big)\geq 1-\varepsilon\Big) \geq 1-\delta,
\end{equation*}
where $C_\tau(x)$ is parameterized by $\tau\in \mathbb{R}$, and the 
$\varepsilon$ and error level $\delta$ are chosen in advance by the user.


%
Technically, the prediction set is constructed by fitting the parameter $\tau$ on the calibration set $\mathcal{D}_\mathrm{cal}$. Existing methods~\cite{park2020pac,khakhar2023pac} rely on the assumption that $\Prob(Y_{n+1}\in C_\tau(X_{n+1}))$ is monotone-decreasing in $\tau$, which is typically guaranteed by the monotonicity of prediction sets, i.e., 
\begin{equation}\label{eq:monotonicity}
\tau<\tau' \rightarrow C_\tau(X_{n+1})\supseteq C_{\tau'}(X_{n+1}).
\end{equation}
Therefore, a standard PAC formulation defines $C_\tau(X_{n+1})$ as the set of all labels whose model-assigned score exceeds the threshold $\tau$:
\begin{equation*}\label{eq:PAC set}
C_\tau(X_{n+1}) := \left\{y\in \mathcal{Y} \mid f(y\mid X_{n+1}) \geq \tau\right\},
\end{equation*}
where $f(y\mid X_{n+1})$ is a scoring function (e.g., the softmax output of a trained classifier).


\subsection{LTT Framework}\label{sec:LTT}

The Learn Then Test (LTT) framework~\cite{angelopoulos2025learn} extends PAC prediction sets by providing control over the expectation of an arbitrary risk function, conditional on the calibration data. 
It achieves this by reframing the selection of the prediction set parameter as a multiple-hypothesis testing problem. 

Specifically, for a parameter $\lambda$, let $\mathcal{T}_{\lambda}(x)\subseteq \mathcal{Y}$ denote the prediction set returned for an input $x\in\mathcal{X}$. 
For a given $\mathcal{T}_{\lambda}$, we define a corresponding {\em risk} $R(\mathcal{T}_{\lambda}) \in \mathbb{R}_{\geq 0}$ that quantifies a problem-specific statistical error rate. 
The goal of the LTT framework is to learn a parameter $\hat{\lambda}$ such that the function $\mathcal{T}_{\hat{\lambda}}$ satisfies the following error-control guarantee for a test instance $(X_{n+1},Y_{n+1})$.
\begin{definition}
[Risk-controlling] A random variable $\hat{\lambda} \in \Lambda$ is risk-controlling at $(\alpha,\delta)$-level, if it satisfies
\begin{equation}\label{eq:RCP}
\Prob( R(\mathcal{T}_{\hat{\lambda}}) \leq \alpha ) \geq 1-\delta.
\end{equation}
\end{definition}
The risk tolerance $\alpha$ and error level $\delta\in (0,1)$ are user-specified. 
For brevity, we abbreviate $R(\lambda)=R(\mathcal{T}_\lambda)$ below. 




To control target risks, the LTT framework leverages the base model and the calibration set $\mathcal{D}_\mathrm{cal}$ to identify valid $\lambda$ via multiple hypothesis testing.
Each $\lambda_k\in \Lambda=\{\lambda_1, \dots,\lambda_N\}$ is associated with the null hypothesis $\mathcal{H}_k:R(\lambda_k)>\alpha$ and $\lambda_k$ is considered valid if $\mathcal{H}_k$ is rejected.
For each $\lambda_k$, a super-uniform $p$-value $p_k$ is computed. 
The output of LTT framework is the reject set 
$\Lambda_{\mathrm{valid}}=\mathcal{A}(p_1,\dots,p_N)\subseteq\Lambda$, 
where $\mathcal{A}$ is any algorithm that controls the family-wise error rate (FWER).
With these steps, we have the following theorem.

\begin{theorem}
[Learn Then Test~\cite{angelopoulos2025learn}]\label{thm:LTT}
Suppose each $p_k$ has a super-uniform distribution under $\mathcal{H}_k$. Let $\mathcal{A}$ be a FWER-controlling algorithm at level $\delta$. Given a test instance $(X_{n+1},Y_{n+1})$, $\Lambda_{\mathrm{valid}}$ satisfies
\begin{equation}\label{eq:thm1}
\Prob \left( \sup_{\lambda\in\Lambda_{\mathrm{valid}}} R(\lambda)  \leq \alpha\mid \mathcal{D}_\mathrm{cal} \right ) \geq 1-\delta,
\end{equation}
where the supremum over an empty set is defined as $-\infty$. 
\end{theorem}

This theorem implies that any $\lambda\in\Lambda_{\mathrm{valid}}$ can be selected to control the risk $R(\lambda)$ on test instances. In particular, we recover a procedure similar to the PAC prediction set by defining $\mathcal{T}_\lambda(X_{n+1}):=\{y\in \mathcal{Y} \mid f(y\mid X_{n+1}) \geq \tau\}$ and $R(\lambda):=\mathbf{1}\{Y_{n+1}\notin\mathcal{T}_\lambda(X_{n+1})\}$.

\section{PAC vs. LTT in Code Generation }\label{sec:vs}

Code generation is a typical structured prediction task with an exponentially large label space, rendering conventional conformal prediction intractable~\cite{vovk2005algorithmic}. 
To address this, \citet{khakhar2023pac} proposed to construct prediction sets via partial programs. 
Their method represents a program as an AST, denoted as $T = (\mathcal{V}, \mathcal{E})$, and formulates an optimization problem which selectively prunes nodes and edges from $T$, yielding a partial program $T_p = (\mathcal{V}_p, \mathcal{E}_p)$. 
The resulting prediction set is defined as the collection of all complete programs that can be obtained by expanding $T_p$, and a correct program $T' = (\mathcal{V}', \mathcal{E}')$ belongs to this set iff $\mathcal{V}_p \subseteq \mathcal{V}'$ and $\mathcal{E}_p \subseteq \mathcal{E}'$. 
This construction 
enables the application of the PAC prediction set framework to code generation, providing statistical guarantees on program correctness. 

However, PAC prediction sets exhibit two limitations when applied to code generation. 
First, they impose a monotonicity constraint (cf.\ Eq.~\eqref{eq:monotonicity}) as a compromise to PAC guarantee~\cite{park2020pac} and can lead to suboptimal partial programs during optimization. 
Second, PAC prediction sets are inherently designed for single-label tasks, whereas code generation naturally admits multiple valid outputs (i.e., multiple syntactically different but semantically-identical programs), rendering it a multi-label problem. 

In contrast, the LTT framework avoids these limitations and is inherently flexible and scalable to multi-label settings. 
Crucially, LTT does not require defining a scoring function for each program in an exponentially large prediction set; it only requires specifying the prediction set itself and its associated risk without any constraint. 
We therefore argue that LTT is better suited for quantifying uncertainty in code generation than PAC prediction sets.
\section{The Approach} \label{sec:method}

We now introduce our method, \name, for constructing structured, risk-controlling prediction sets for LLM-based code generation. At a high level, \name defines a structured risk function on code prediction sets via solving an optimization problem.
For each calibration program $X_i\in\mathcal{D}_\mathrm{cal}$, we sample multiple correct programs to obtain the enhanced label set $\tilde{Y}_i$, according to the corresponding generation task of $X_i$.
Finally, we apply multiple hypothesis testing, following the procedure of LTT, to identify valid parameters $\Lambda_{\mathrm{valid}}$. These parameters are then used to generate risk-controlling prediction sets for the test program $X_{n+1}$.
Additionally, we propose a selective execution strategy to lower the computational cost of verifying sample programs' correctness, by tolerating a controlled rate of labeling errors.





\subsection{Structured Risk Function}\label{sec:risk}
Unlike standard multi-label classification problems where the false discovery rate (FDR) is a common control target, our focus is on the semantic correctness of the structured output. 
Specifically, given an input program $X$, we require the prediction set $\mathcal{T}_\lambda(X)$ to contain at least one program that belongs to the set $\tilde{Y}$ of correct programs. 
Formally, we define the empirical structured risk on the enhanced calibration set $\{(X_i,\tilde{Y}_i)\}_{i=1}^n$ as  
\begin{equation}\label{eq:risk}
\widehat{R}_n^s(\lambda):=\frac{1}{n}\sum_{i=1}^n\mathbf{1}\{\mathcal{T}_\lambda(X_i)\cap \tilde{Y}_i=\emptyset\}. 
\end{equation}
%
%

Since the label space in code generation is high-dimensional and exponentially large, directly enumerating all complete programs is intractable.
Following \citet{khakhar2023pac}, we construct the prediction set $\mathcal{T}_\lambda(X)$ using partial programs, which provide a compact representation of the candidate program set, as mentioned in Section~\ref{sec:vs}.


The remaining problem is to obtain partial programs for $X_i$ accounting for uncertainty in code generation. Given an AST $T=(\mathcal{V},\mathcal{E})$ and a parameter $\lambda_k\in\Lambda$, we introduce a binary variable $\beta^{(k)}_{i,v}$ for each node $v\in\mathcal{V}$, where $\beta^{(k)}_{i,v}=1$ indicates that $v$ is removed from $T$, and $\beta^{(k)}_{i,v}=0$ otherwise.
Let $\pi_v\in\mathbb{R}_{\geq 0}$ denote the negative log probability of $v$ conditioned on its ancestors, i.e., $\pi_v=-\log \Prob(v\mid v_1,\dots,v_m)$.
Partial programs can be obtained by 
optimization. For fixed $k\in[N]$ and input program $X_i$:
\begin{equation}\label{eq:partial}
\min_{\beta} \sum_{v\in \mathcal{V}} \beta^{(k)}_{i,v}, ~~ \mathrm{s.t.} \sum_{v\in \mathcal{V}} \pi_{v}(1-\beta^{(k)}_{i,v}) \leq \lambda_k. 
\end{equation}

In this minimization problem, the goal is to retain as many nodes as possible while ensuring the total uncertainty of the partial programs does not exceed $\lambda_k$.
Moreover, the solution must satisfy two additional structure constraints derived from ASTs:
\begin{equation}\label{eq:tree constraint}
\begin{aligned}
&\mathrm{(SC1) } ~~(\beta^{(k)}_{i,v} = 1) \rightarrow (\beta^{(k)}_{i,v'} = 1), &&\forall (v,v') \in \mathcal{E}, \\
&\mathrm{(SC2) } ~\sum_{(v, v') \in \mathcal{E}} \beta^{(k)}_{i,v'}(1-\beta^{(k)}_{i,v}) \leq t_{\max},
\end{aligned}
\end{equation}
where (SC1) enforces tree connectivity and (SC2) imposes an upper-bound $t_{\max}\in\mathbb{N}^+$ on the number of permissible subtree removals.

In contrast to PAC prediction set~\cite{khakhar2023pac}, our optimization problem for constructing partial programs operates under a less restrictive set of constraints. The most significant relaxation is the removal of the monotonicity condition on prediction sets (cf. Eq.\eqref{eq:monotonicity}), which substantially enlarges the feasible label space.



\subsection{Multiple Hypothesis Testing for Code}\label{sec:mht}
 
With the structured risk defined above, we then employ multiple hypothesis testing to identify the reject set $\Lambda_{\mathrm{valid}}$.
The first step of this procedure is to compute a $p$-value $p_k$ for each null hypothesis $\mathcal{H}_k$. 
Since the structured risk defined in Eq.\eqref{eq:risk} is an average of binary losses, 
a valid $p$-value can be computed by the following lemma. 
\begin{lemma}[Binomial trial $p$-values]\label{lemma:p-value}
Let $\widehat{R}_n^s(\lambda_k)$ be the empirical structured risk defined in Eq.\eqref{eq:risk}, and let $\mathrm{Bin}(n,\alpha)$ denote a binomial random variable with $n$ trials and success probability $\alpha$. Then,
\begin{equation}\label{eq:p-value}
p_k^{BT} = e\Prob\big(\mathrm{Bin}(n,\alpha)\leq n\widehat{R}_n^s(\lambda_k)\big)
\end{equation}  
is a valid p-value for $\mathcal{H}_k:R^s(\lambda_k)>\alpha$.
\end{lemma}


In LTT, let $\mathcal{A}$ be any FWER-controlling algorithm at level $\delta$.
Applying $\mathcal{A}$ to the $p$-values $\{p^{BT}_k\}_{k=1}^N$ yields a set of rejected hypotheses, whose corresponding configurations form $\Lambda_{\mathrm{valid}}\subseteq\Lambda$.
If $\Lambda_{\mathrm{valid}}=\emptyset$, the method abstains (i.e., returns \texttt{null}).
Otherwise, any parameter in $\Lambda_{\mathrm{valid}}$ can be selected.
In practice, we select the one that empirically minimizes the number of node removals.
Since the number of removals decreases monotonically with $\lambda$, this amounts to choosing
\begin{equation}\label{eq:lambda}
\hat{\lambda}=\max(\Lambda_{\mathrm{valid}}).
\end{equation}
Note that the monotonicity here only affects the number of pruned nodes (i.e., the size of prediction sets), which naturally follows from the optimization objective in Section~\ref{sec:risk}.
This is distinct from and weaker than the explicit monotonicity constraint on prediction sets imposed in Eq.~\eqref{eq:monotonicity}.
Consequently, the parameter $\hat{\lambda}$ selected by Eq.~\eqref{eq:lambda}, based on $\mathcal{D}_\mathrm{cal}$, provides a risk-controlling guarantee for the resulting prediction sets.

\begin{algorithm} 
\caption{Risk-Controlling Partial Program}
\label{alg:LTT}
\begin{algorithmic}[1]
\STATE {\bfseries Input:} enhanced calibration set $\mathcal{D}_\mathrm{cal}=\{(X_i,\tilde{Y}_i)\}_{i=1}^n$; test program $X_{n+1}$; parameters $\Lambda=\{\lambda_k\}_{k=1}^N$;  maximum of subtree removals $t_{\max}$; risk tolerance $\alpha$; error level $\delta$;
\STATE {\bfseries Output:} partial program $\beta_{n+1,v}$ for $X_{n+1}$.
\FOR {$k=1,\dots,N$}
\FOR {$i=1,\dots,n$}
\STATE $\beta^{(k)}_{i,v}\leftarrow$ Solve  Eqs.\eqref{eq:partial}-\eqref{eq:tree constraint} for $X_i$, given $\lambda_k$;
\ENDFOR
\STATE Define $\widehat{R}_n^s(\lambda_k)$ using Eq.~\eqref{eq:risk};
\STATE Compute p-value $p_k^{BT}$ for $\widehat{R}_n^s(\lambda_k)$ using Eq.~\eqref{eq:p-value};
\ENDFOR
\STATE Obtain $\Lambda_{\mathrm{valid}}=\mathcal{A}(p_1^{BT},\dots,p_N^{BT})$ via controlling FWER;
\STATE Select $\hat{\lambda}=\max(\Lambda_{\mathrm{valid}})$;
\STATE $\beta_{n+1,v}\leftarrow$ Solve  Eqs.\eqref{eq:partial}-\eqref{eq:tree constraint} for $X_{n+1}$, given $\hat{\lambda}$.
\end{algorithmic}
\end{algorithm}


Algorithm~\ref {alg:LTT} summarizes how to obtain risk-controlling partial programs via picking desirable $\hat{\lambda}$. 
This result ensures that the prediction set $\mathcal{T}_{\hat{\lambda}}(X_{n+1})$ constructed by these partial programs is also risk-controlling.

\begin{theorem}[Structured risk-controlling sets]\label{thm:guarantee}
The prediction set $\mathcal{T}_{\hat{\lambda}}(X_{n+1})$, constructed by the partial program from Algorithm~\ref{alg:LTT}, satisfies Eq.~\eqref{eq:RCP}.
\end{theorem}

To precisely search and test $\Lambda$ in the context of code generation, in Appendix~\ref{A:FWER}, we discuss three alternative FWER control algorithms: \textit{Bonferroni method}, \textit{Holm–Bonferroni method}~\cite{holm1979simple}, and \textit{fixed sequence testing}~\cite{bauer1991multiple}. 

\subsection{Selective Execution Strategy}\label{sec:sv}
As mentioned in Section~\ref{sec:vs}, code generation is inherently a multi-label problem. The LTT framework supports label augmentation on the calibration set $\mathcal{D}_\mathrm{cal}$. 
For each generation task of $X_i\in\mathcal{D}_\mathrm{cal}$, we sample additional $m$ programs from the generative code model and execute test cases to identify correct instances. 
Such a sampling and verification step is performed only during calibration, incurring no computational overhead at test time. 

Despite this, in scenarios with severely limited computational resources, we propose a selective execution strategy that maintains high-quality labels while improving sampling efficiency. 
Given a total of $M$ sampled programs on $\mathcal{D}_\mathrm{cal}$, we use program uncertainty scores $\{U_i\}_{i=1}^M$ (e.g., perplexity) to decide whether executing test cases is necessary.
We learn a threshold $\hat{u}$ such that programs with $U_i < \hat{u}$ are directly accepted as correct, while those with $U_i \geq \hat{u}$ are verified by executing test cases. 
This strategy does not require perfectly calibrated models, 
but better calibration can indeed enhance the precision of $\hat{u}$. 
Formally, we define the cumulative error function conditioned on $u$ as
\begin{equation*}
L(u)=\frac{1}{M}\sum_{i=1}^{M}\ell(S_i)\mathbf{1}\{U_i<u\},
\end{equation*}
where $S_i\in\{0,1\}$ is a true indicator ($1$ for a correct program, $0$ otherwise), and $\ell(S_i)=1-S_i$ computes the error.

Our goal is to select a threshold $u$ such that $L(u)\leq\epsilon$, where $\epsilon$ is a label error rate.
However, directly computing $L(u)$ by verifying all $M$ sampled programs is computationally expensive.
To circumvent this, we instead estimate an error upper bound $\hat{L}^u(\gamma)$ on $L(u)$ that holds  pointwise with probability $\gamma$, i.e., satisfies the following inequality:
\begin{equation}\label{eq:error upper bound}
\Prob(\hat{L}^u(\gamma) \geq L(u)) \geq 1-\gamma.
\end{equation}

To compute an approximate error bound $\hat{L}^u(\gamma)$, we propose a procedure analogous to importance sampling, inspired by \citet{candes2025probably}.
We begin by drawing $h$ indices $\{i_1,\dots,i_h\}$ uniformly with replacement from $[M]$. 
For each sampled index $i_j$, whether a test-case execution is performed is determined stochastically by a Bernoulli trial $\xi_{i_j}\sim \mathrm{Bern}(\omega_{i_j})$, where $\{\omega_i\}_{i=1}^M$ are predefined sampling weights.
This procedure yields a set $\{Z_j(u)\}_{j=1}^h$ of $h$ i.i.d. random variables, defined as $Z_j(u)=\ell(S_{i_j})(\xi_{i_j}/\omega_{i_j})\mathbf{1}\{U_{i_j}\leq u\}$. 


Since $\E[\xi_{i_j}/\omega_{i_j}\mid i_j]=1$, the expectation of $Z_j(u)$ equals the target quantity $L(u)$. Therefore, an unbiased estimate of $L(u)$ can be obtained from the sample mean of $\{Z_j(u)\}_{j=1}^h$. An upper confidence bound for $L(u)$ is then 
constructed 
by Hoeffding’s inequality~\cite{hoeffding1963probability}, yielding a finite-sample upper bound
\begin{equation}\label{eq:hoeffding bound}
\hat{L}^u(\gamma)=\hat{\mu}_Z(u) + \delta_{HB}(\gamma),
\end{equation}
where $\hat{\mu}$ denotes the empirical mean of $\{Z_j(u)\}_{j=1}^h$ and Hoeffding bound $\delta_{HB}(\gamma)=\sqrt{\log(2/\gamma)/2h}/\omega_{\mathrm{min}}$, where $\omega_{\mathrm{min}}=\min\{\omega_i\}_{i=1}^M$.

After obtaining $\hat{L}^u(\gamma)$, let $\hat{u} = \max\{u\in[0,1]:\hat{L}^u(\gamma)\leq\epsilon\}$ where $\epsilon\in(0,1)$, and the final predicted label $\hat{S}_i$ is
\begin{equation*}
\hat{S}_i=S_i\mathbf{1}\{U_i\geq \hat{u}\} + \mathbf{1}\{U_i< \hat{u}\}.
\end{equation*}

Algorithm~\ref{alg:EUB} describes the selective execution procedure.
We have the following error upper bound.
\begin{lemma}[Non-asymptotic error upper bound]\label{lemma:Hbound}
If the label set $\{\hat{S}_i\}^M_{i=1}$ of sample programs is computed by Algorithm~\ref{alg:EUB}, then it meets Eq.~\eqref{eq:error upper bound} and $\hat{L}^u(\gamma)\leq \epsilon$.
\end{lemma}


This result provides a non-asymptotic upper bound for $L(u)$, suitable for settings with limited sample sizes (small $M$). 
Alternatively, an asymptotic bound for large $M$ is given in Appendix~\ref{A:CLT bound}.

Combining Lemma~\ref{lemma:Hbound} and Theorem~\ref{thm:guarantee}, we obtain the error-controlling guarantee when using the selective execution strategy.

\begin{proposition}[Approximately risk-controlling sets]\label{prop:sampling}
Let the output of Algorithm~\ref{alg:EUB} be the calibration set $\mathcal{D}_\mathrm{cal}$. 
Then, the parameter $\hat{\lambda}$ selected by Eq.~\eqref{eq:lambda} satisfies
\begin{equation*}
\Prob\big(R^s(\hat{\lambda})\leq \alpha+\epsilon(1-\alpha)\big)\geq (1-\gamma)\cdot(1-\delta).
\end{equation*}

\end{proposition}

\begin{remark}
This guarantee presents a practical trade-off, allowing for controlled relaxation of the risk guarantee in exchange for computational savings by permitting a non-zero label error rate $\epsilon$.
When $\epsilon=0$,
the result reduces to the original guarantee of Theorem~\ref{thm:guarantee}.
\end{remark}


\begin{algorithm} 
\caption{Selective Execution with Error Bound} 
\label{alg:EUB}
\begin{algorithmic}[1]
\STATE {\bfseries Input:} unlabeled sampling programs $\{X_i\}_{i=1}^M$; uncertainty scores $\{U_i\}_{i=1}^M$; test cases $\{C_i\}_{i=1}^M$; sampling weights $\{\omega_i\}_{i=1}^M$; sampling size $h$; label error rate $\epsilon$; confidence level $\gamma$;
\STATE {\bfseries Output:} labeled sampling programs $\{X_i,\hat{S}_i\}^M_{i=1}$.
\FOR {$j=1,\dots,h$}
\STATE Sample an index $i_j\sim \mathrm{Unif}([M])$;
\STATE Sample a random variable $\xi_{i_j}\sim \mathrm{Bern}(\omega_{i_j})$;
\IF{$\xi_{i_j}=1$}
\STATE Execute the test case $C_{i_j}$ to get $S_{i_j}$;
\STATE Compute the importance-weighted error using $Z_j=\ell(S_{i_j})/\omega_{i_j}$;
\ELSE
\STATE $Z_j=0$;
\ENDIF
\ENDFOR
\STATE Given a threshold $u$, define $Z_j(u)=Z_j\cdot\mathbf{1}\{U_{i_j}\leq u\}$ for $j\in[h]$;
\STATE Compute error upper bound $\hat{L}^u(\gamma)$ on $\{Z_j(u)\}_{i=1}^h$ for all $u\in\{U_i\}_{i=1}^M$ using Eq.~\eqref{eq:hoeffding bound};
\STATE Determine $\hat{u} = \max\{u\in[0,1]:\hat{L}^u(\gamma)\leq\epsilon\}$;
\STATE Obtain $\hat{S}_i=S_i\mathbf{1}\{U_i\geq \hat{u}\} + \mathbf{1}\{U_i< \hat{u}\}$ for $i\in[M]$.
\end{algorithmic}
\end{algorithm}

\section{Experiments} \label{sec:exp}


\noindent\textbf{Datasets.}
We evaluate our approach on three standard Python code generation benchmarks: HumanEval~\cite{chen2021evaluating}, MBPP~\cite{austin2021program} and APPS~\cite{hendrycks2021measuring}. HumanEval consists of 164 hand-written programming problems. 
We use the sanitized version of the MBPP dataset (nearly 500 problems) as provided by CodeT~\cite{chen2022codet}, and select the introductory-level problems of the APPS dataset (1,000 problems). 
Each dataset is then randomly partitioned into calibration and test sets using a 50/50 split. 
To mitigate variance from random partitioning, we repeat this process over 100 independent random trials.

\noindent\textbf{Baselines.}
We compare our method \name against two state-of-the-art baselines: the PAC prediction set approach~\cite{khakhar2023pac} and its greedy variant. 
The PAC baseline constructs partial programs by solving a complex optimization problem, while the greedy baseline iteratively removes the AST node that reduces the partial program’s negative log-likelihood (NLL) the most.

\noindent\textbf{Evaluation metrics.}
To evaluate the structured prediction sets $\mathcal{T}_{\hat{\lambda}}(X_{n+1})$ produced by our \name, we report two metrics: (1) the percentage of nodes removed from AST, and (2) the coverage of code sets that contain at least one correct program. 
We aim to keep node removals small while ensuring that coverage remains above the target threshold $1-\alpha$, thereby preserving more useful program information.

\noindent\textbf{Hyperparameters.}
To control risk, the key hyperparameter is the search space for $\lambda$. 
We implement a uniform grid search with an increment of 0.02, a simple but effective discretization. 
Empirically, we observe that outcomes are insensitive to the exact resolution, as long as the grid density is adequate to capture the operative range of $\lambda$.
We set the maximum subtree removals $t_{\max}=1$ and the additional sampling quantity $m=20$ for label augmentation on the calibration set by default. 
For the risk parameter, we fix $\delta = 0.1$ and vary $\alpha$ in our experiments.

\noindent\textbf{Implementations.}
We evaluate three open-source LLMs: the code-specialized Deepseek-Coder-33B~\cite{guo2024deepseek} and Qwen2.5-Coder-32B~\cite{hui2024qwen2}, along with the general-purpose Llama3.1-70B~\cite{grattafiori2024llama}. 
For all LLMs, we use a sampling temperature of 0.8 and a top-$p$ of 0.95 by default. 
We choose fixed sequence testing as the FWER algorithm for LTT.
Moreover, we solve the optimization problem for generating partial programs using the Z3 solver, enforcing a 24‑second timeout per instance in accordance with SMT‑COMP\footnote{SMT‑COMP is the annual international competition 
for Satisfiability Modulo Theories (SMT) solvers.} rules.
All the experiments are carried out with Python 3.9.19, Pytorch 2.4.0 and vLLM 0.6.1, running on NVIDIA H800 GPUs with CUDA 12.7.

\begin{figure*}[t]
 \centering
 \subfigure[DeepSeek-Coder on HumanEval\label{F:deepseek}]{
   \begin{tabular}{c}
     \includegraphics[width=1.9in]{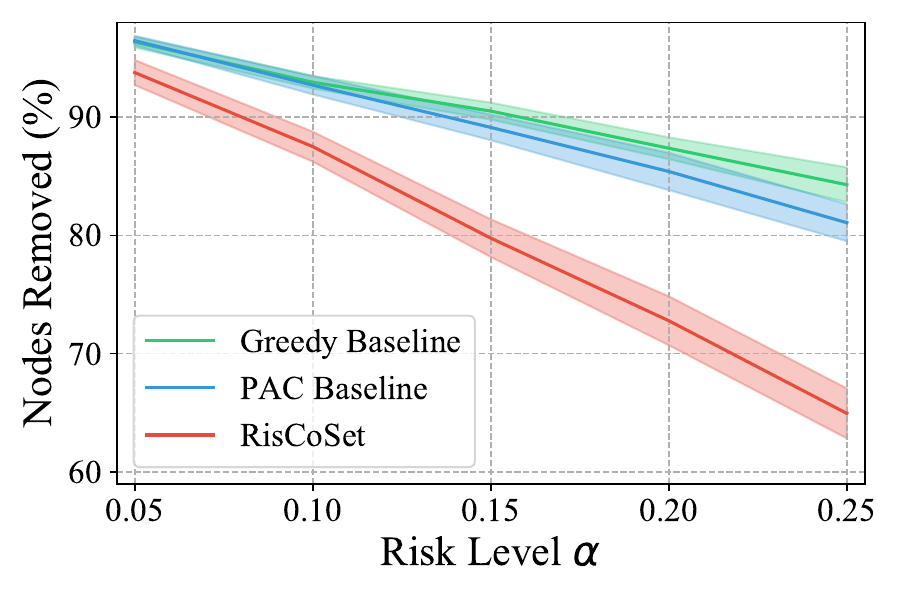} \\
     \vspace{-1em} \\
     \includegraphics[width=1.9in]{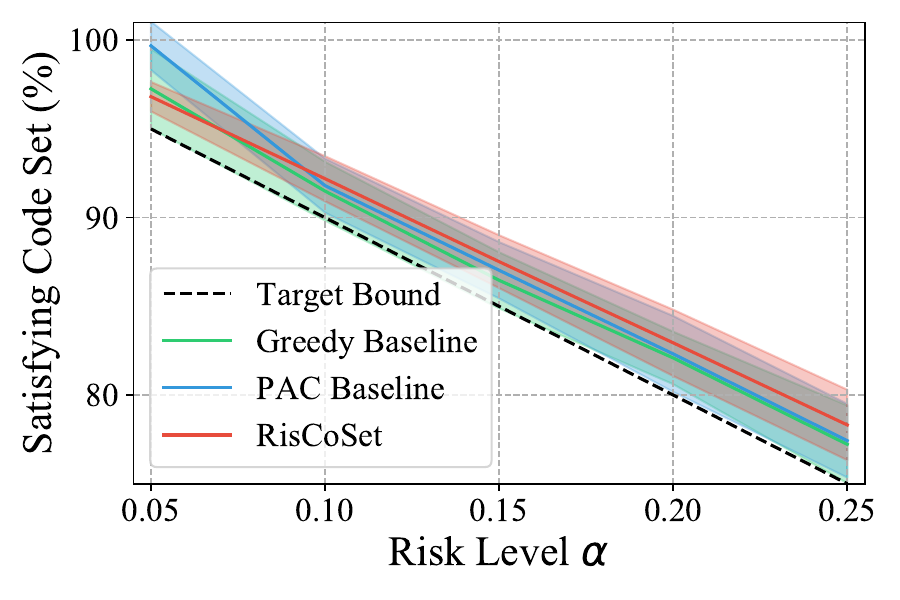}
   \end{tabular}
 }
 \subfigure[Qwen2.5-Coder on MBPP\label{F:qwen}]{
   \begin{tabular}{c}
     \includegraphics[width=1.9in]{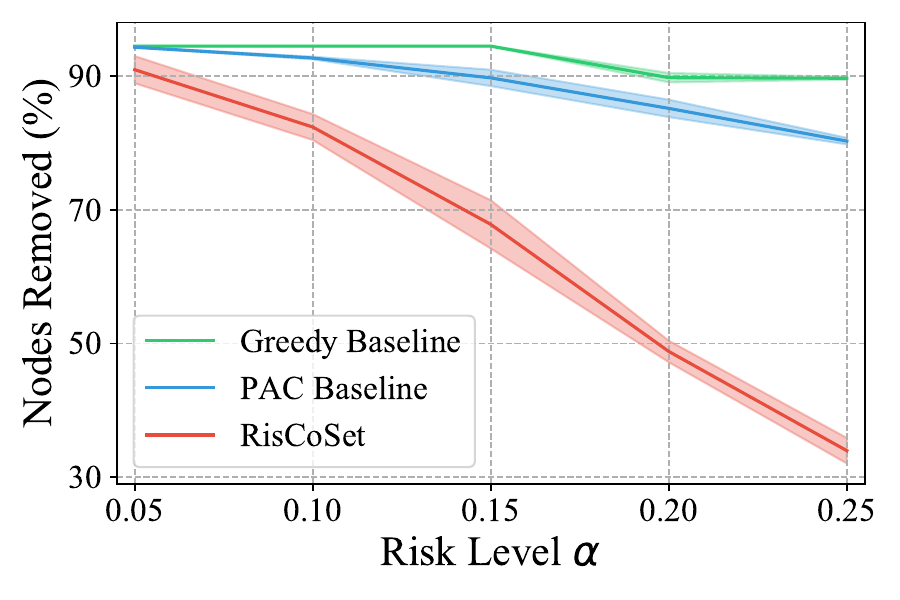} \\
     \vspace{-1em} \\
     \includegraphics[width=1.9in]{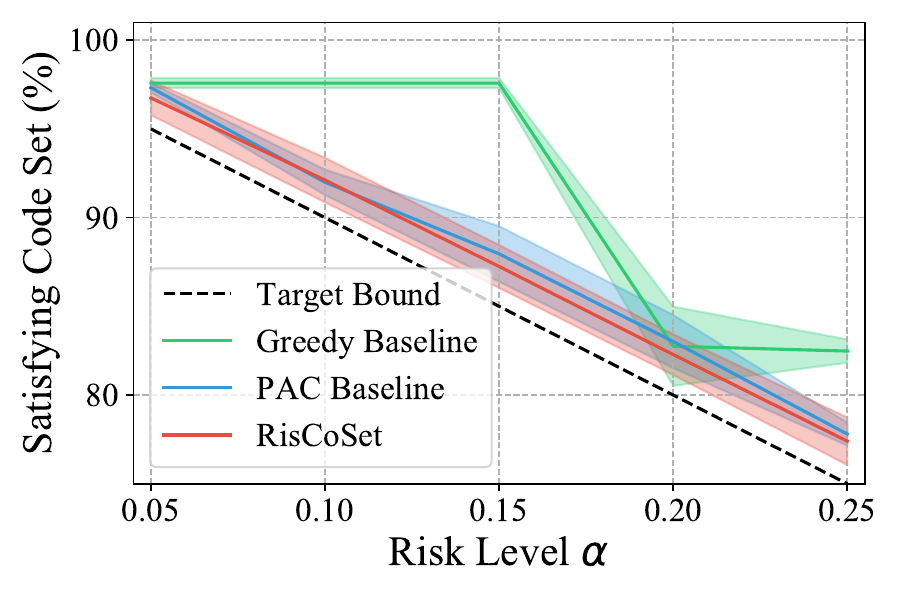}
   \end{tabular}
 }
 \subfigure[Llama3.1 on APPS\label{F:llama}]{
   \begin{tabular}{c}
     \includegraphics[width=1.9in]{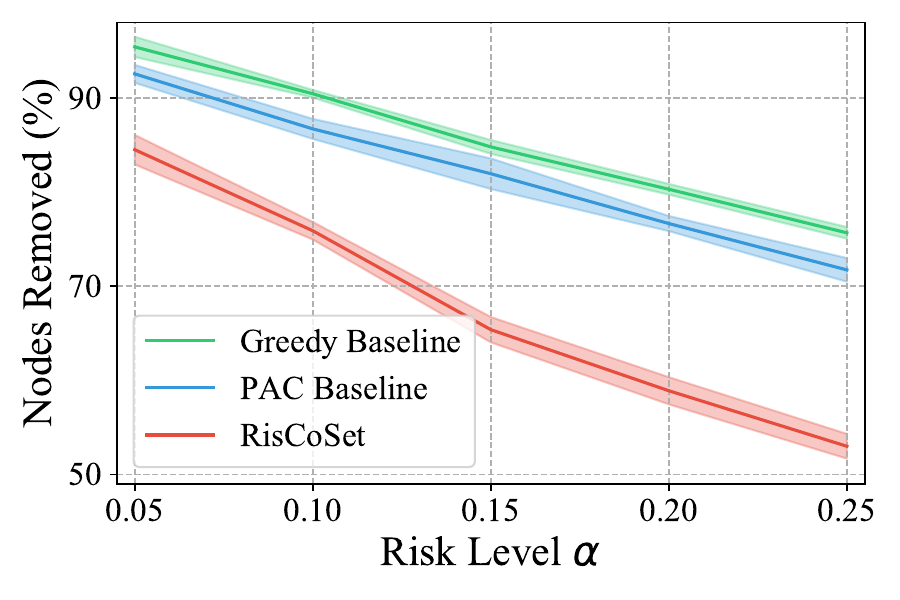} \\
     \vspace{-1em} \\
     \includegraphics[width=1.9in]{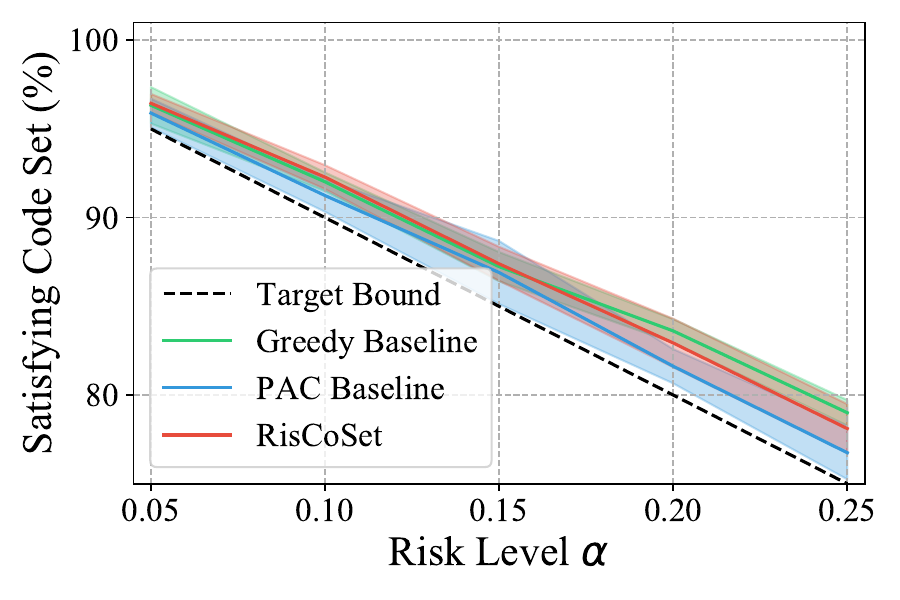}
   \end{tabular}
 }
 
 \caption{
 Percentage of node removals (top row) and satisfying code sets (bottom row) w.r.t. risk level $\alpha$. The results are the mean over 100 random splits.
 The smaller node removal is better, when the code set coverage exceeds the target bound $1-\alpha$.
 Our approach constructs prediction sets that remove significantly fewer nodes compared to baselines for three LLMs on all datasets.}\label{F:exp1}
 \vspace{-1.0em}
\end{figure*}

\begin{figure*}[t]
\begin{minipage}{0.63\linewidth}
\centering
\subfigure{\label{F:m-remove}
    \includegraphics[width=2.0in]{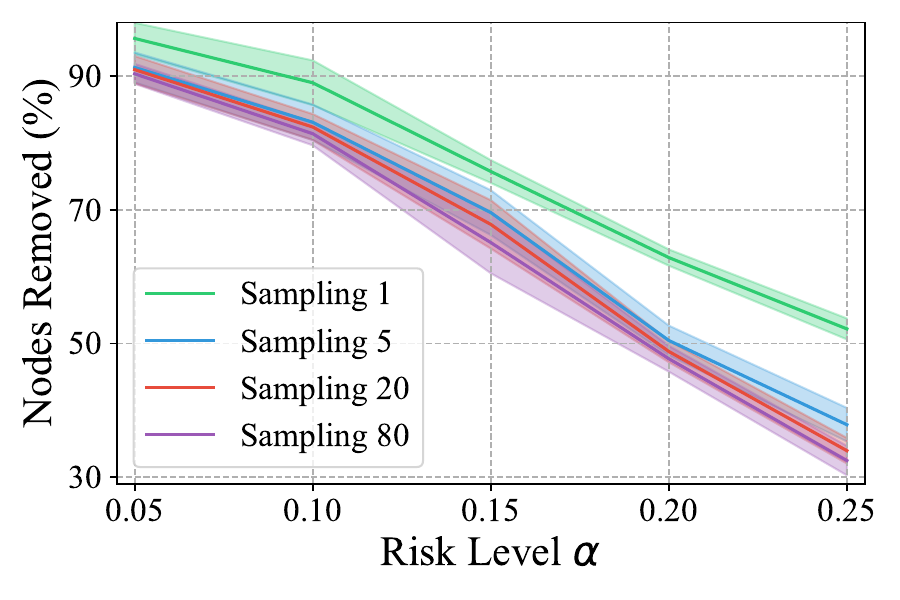}}
\subfigure{\label{F:m-set}
     \includegraphics[width=2.0in]{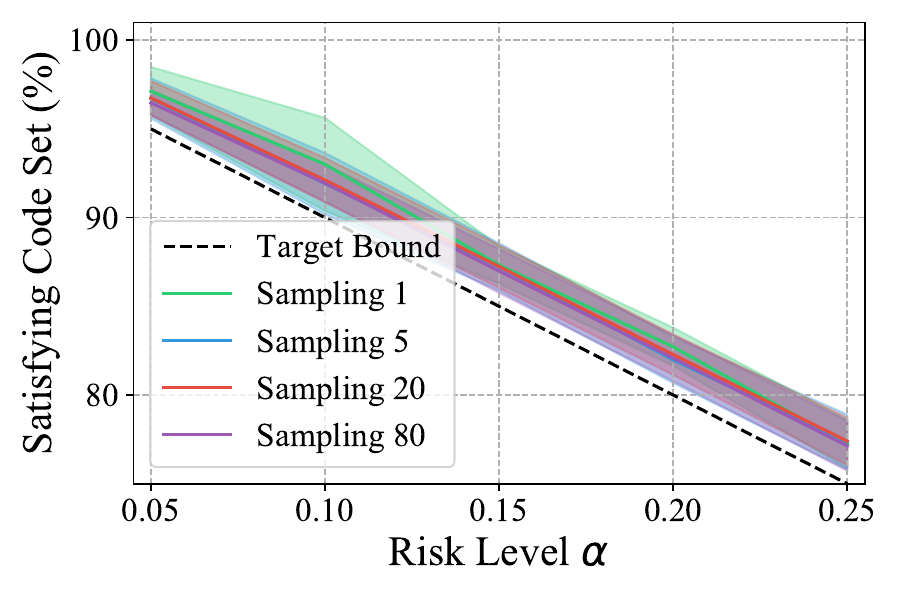}}
\caption{Parameter sensitivity analysis of sampling quantity $m$ w.r.t. risk level $\alpha$ on MBPP. 
The average results over 100 trials show that a larger value of $m$ leads to fewer node removals, while maintaining the required risk control.}
\label{F:class condition}
\end{minipage}
\hfill
\begin{minipage}{0.34\linewidth}
\centering
\vspace{1.4em}
\resizebox{\linewidth}{!}{
\begin{sc}
\begin{tabular}{lcccc}
\toprule
{Error $\epsilon$} & 0.05 & 0.1  & 0.2 & 0.3 \\
 \midrule
Removal 
& 71.8 & 50.6 & 27.8 & 3.39 \\
\midrule
\multirow{2}*{Coverage}
& 88.5 & 82.9 & 74.1 & 66.5 \\
& (85.5) & (81.0) & (72.0) & (63.0) \\
 \midrule
Save& 12.0 & 29.4 & 57.6 & 84.4 \\
\bottomrule
\end{tabular}
\end{sc}}
\vspace{5pt}
\captionof{table}{Percentage of node removals (\%), coverage of satisfying code set (\%), and fraction of calibration data saved from execution (\%) w.r.t. label error $\epsilon$ when $\alpha=0.1$ on MBPP.
The average results over 100 trials show that the coverage of the satisfying code set is above the target coverage (i.e., the values in ()), and our selective execution strategy effectively reduces the proportion of programs needed to execute.}
\label{T:selective verification}
\end{minipage} 
\vspace{-1.5em}
\end{figure*}


\subsection{Results}
Our main results are presented in Fig.~\ref{F:exp1}. 
For all plots, the solid line shows the average values over 100 trials, and the shaded regions indicate the standard deviation. Additional experimental results are reported in Appendix~\ref{A:exp}. 

\noindent\textbf{Validity of risk-controlling code generation.}
We implement our approach \name for different LLMs on three datasets.
The results in Fig.~\ref{F:exp1} illustrate, with varying risk levels $\alpha$, the percentage of nodes removed from the pruned AST (top row) and the coverage of constructed prediction sets that contain a correct program (bottom row). 

Our empirical findings align with Theorem~\ref{thm:guarantee} in practice, as the average code set coverage consistently exceeds the target bound of $1-\alpha$ (the dashed lines).
Across all datasets, our method produces prediction sets that remove significantly fewer AST nodes than the baselines. 
For instance, as shown in Fig.~\ref{F:qwen} and Fig.~\ref{F:llama}, our approach reduces node removals by 24.5\% and 20.3\% compared to the best competitor when $\alpha=0.15$ on the MBPP and APPS datasets, respectively.
Compared to the PAC baseline, the greedy baseline tends to remove more nodes, which is particularly evident on the MBPP dataset in Fig.~\ref{F:qwen}, due to its greedy pruning strategy.



For the code-specialized models, our approach performs better with Qwen2.5-Coder than with Deepseek-Coder. 
This may be attributed to Qwen2.5-Coder's stronger code generation capability~\cite{hui2024qwen2}, suggesting that our method can benefit increasingly from more powerful generative models. 
Moreover, the proposed method also achieves performance comparable to that of specialized code generators when applied to the general-purpose Llama3.1 in Fig.~\ref{F:llama}, showing its scalability for diverse LLMs.

The significant improvement of our method can be attributed to: (1) we define an optimization problem with fewer constraints, encouraging a better partial program solution with more remaining nodes; (2) our multi-label formalization allows for a more complete exploration of the potential correct program space for a specific LLM compared to baselines.

\noindent\textbf{Parameter sensitivity analysis.}
We next analyze the sensitivity of our method to the sampling quantity 
$m$. 
We sample $m$ additional programs for each task in the calibration set with $m=\{1,5,20,80\}$, and pick correct ones as new labels via verification.
The corresponding node removals and code set coverage on the MBPP dataset
are given in Fig.~\ref{F:m-remove} and Fig.~\ref{F:m-set}. 
We observe a clear trend that a larger $m$ helps reduce the number of pruned AST nodes while preserving the desired risk control. 
Insufficient exploration of the label space at $m=1$ leads to more node removals. 
Nevertheless, the performance gains diminish for $m> 5$, demonstrating that our method already performs well with a relatively small sample size. 

\noindent\textbf{Selective execution strategy.}
We evaluate the effectiveness of the proposed selective execution strategy when $\alpha=0.1$ and $m=80$ on the MBPP dataset.
In experiments, we set confidence level $\gamma=0.01$, sampling size $h$ as 10\% of total sample quantity $M$, and samping weights $\omega=\omega_i=1.0$, for any $i\in[M]$.
We report the average results with varying label error $\epsilon$ over 100 trials in Table~\ref{T:selective verification}.
The key metric for efficiency is the percentage of sample programs saved from execution at the cost of $\epsilon\%$ error in program labels.
First, the observed coverage of the satisfying code set always exceeds the target coverage, which controls risks, consistent with the theoretical guarantee in Proposition~\ref{prop:sampling}.
Second, our selective execution strategy largely reduces the proportion of verifications by running test cases. 
For instance, at $\epsilon=0.1$, it avoids test-case execution for nearly 30\% of generated programs, substantially lowering computational cost.

\section{Related Work} \label{sec:rel}

\textbf{Conformal prediction and risk control.}
\hide{
Uncertainty Quantification (UQ) aims to characterize predictive uncertainty in machine learning models~\cite{abdar2021review}. 
There exist several conventional UQ methods,
including Bayesian and approximate Bayesian methods~\cite{gal2016dropout}, post-hoc confidence calibration techniques~\cite{guo2017calibration}, and model-agnostic predictive intervals~\cite{romano2019conformalized}.
However, most of these methods lack rigorous statistical guarantees on coverage, particularly under non-i.i.d. conditions.
}
Conformal Prediction (CP)~\cite{vovk2005algorithmic} is proposed to provide a framework for uncertainty quantification with distribution-free coverage guarantees~\cite{angelopoulos2023conformal}. 
CP has been applied to a range of tasks, including image classification~\cite{sadinle2019least}, object detection~\cite{teng2022predictive}, and regression prediction~\cite{kato2023review}. 

These traditional CP methods, designed for classification or regression, are ill-suited for generative tasks due to the finite output space.
Therefore, some risk-controlling frameworks are proposed for generation tasks~\cite{bates2021distribution,angelopoulos2022conformal}.
However, these methods, such as PAC prediction sets~\cite{park2020pac}, require an additional monotonicity constraint regarding risk, making them difficult to adapt to multi-label problems.
To this end, the Learn Then Test (LTT) framework~\cite{angelopoulos2025learn} separates model training from statistical validation and uses multiple hypothesis testing to control risk. 
Without imposing the monotonicity assumption on the output space, LTT is particularly applicable to multi-label settings like code generation.



\textbf{Uncertainty estimation on LLMs.} 
Existing work~\cite{ Kadavath2022LanguageM, vasconcelos2023generation} has shown that the probability scores produced by LLMs may be poorly calibrated. 
Current uncertainty calibration methods for LLMs mainly focus on prompting~\cite{lin2022teachingmodelsexpressuncertainty} or post-hoc logit calibration (e.g., temperature scaling~\cite{guo2017calibration} and self-consistency~\cite{zhou2025a}), rather than providing theoretical guarantees.

\hide{
Most existing approaches to uncertainty estimation for LMs focus on confidence estimation rather than formal guarantees. One class of methods introduces prompting strategies or linguistic cues to elicit self-reported confidence from the model~\cite{lin2022teachingmodelsexpressuncertainty, tian2023just, xiong2023can}. Another line of work applies post-hoc calibration techniques to model logits~\cite{wang2025understandingcharacteristicscodegeneration, huang2023look}, including temperature scaling~\cite{guo2017calibration} and self-consistency-based confidence calibration~\cite{zhou2025a, kuhn2023semantic, cheng2024relic, mitchell2022enhancingselfconsistencyperformancepretrained, wang2023selfconsistencyimproveschainthought} with the goal of improving probabilistic calibration without modifying the underlying model. 
}

Additionally, recent work has investigated the application of CP to LLMs to achieve rigorous statistical guarantees on the correctness of generated outputs. 
\hide{
Conformal inference provides a general framework for transforming the predictions of an arbitrary model into prediction sets that contain the true outcome with high probability, without distributional assumptions. Several recent studies leverage this framework to ensure that, with high confidence, at least one element in a conformal set of LLM-generated responses satisfies a desired property, such as factual correctness or validity, by defining task-specific nonconformity scores over generated outputs~\cite{kumar2023conformalpredictionlargelanguage, ye2024benchmarkingllmsuncertaintyquantification}. Recent work extends conformal inference for language models from response-level selection to identifying subsets of generated text that satisfy high-probability correctness guarantees, addressing limitations in conditional validity and utility through adaptive conformal procedures and improved scoring functions~\cite{mohri2024languagemodelsconformalfactuality, cherian2024largelanguagemodelvalidity}. 
}
One line of work applies it to produce a set of candidate responses that, with high probability, contains a valid answer~\cite{kumar2023conformal,ye2024benchmarking}. 
This task‑level approach, however, has limited practicality because evaluating a large number of candidates is computationally expensive. 
Alternatively, sample‑level methods apply conformal prediction internally to a single LLM output, identifying its sound components as a prediction set~\cite{quach2024conformallanguagemodeling,cherian2024largelanguagemodelvalidity}.
However, these approaches rely on the assumption of near-independence among components. While this may hold in NLP contexts, it breaks down for structured domains like code, where components are systemically interdependent~\cite{casalnuovo2019studying,chen2022codet}.

\textbf{Reliable code generation.} 
Ensuring code correctness has long been a central concern in software engineering. 
Conventional methods include
testing (e.g., unit testing and fuzzing)~\cite{2022fuzzing}, static analysis~\cite{2021parkasurvey}, and formal verification~\cite{Baier2008PrinciplesOM}. 
Nevertheless, such methods are inadequate for LLM‑generated code, as they cannot effectively leverage the probabilistic nature and uncertainty of modern language models.

The advancement of LLMs for code generation has spurred interest in reliability methods based on statistical, rather than purely deductive, guarantees~\cite{chen2022codet}.
For example, \citet{guo2021learning} propose to complete the uncertain code parts within the program
sketches to improve the pass rate.
Recently, research has explored constructing PAC prediction sets for code generation~\cite{khakhar2023pac}. 
By representing programs via ASTs and building prediction sets defined by partial programs, this method delivers probabilistic guarantees that the code set contains the ground-truth program.
Different from the previous work, our approach introduces LTT framework to eliminate the need for monotonicity constraints and naturally accommodates the multi-label essence of code.

\section{Conclusion} \label{sec:con}

This paper addresses the challenge of unreliable LLM outputs in structured generation tasks like code. 
We introduce a method for constructing statistically guaranteed prediction sets by reformulating the problem as multi-label classification, eliminating monotonicity constraints. 
Our approach applies the LTT framework to ensure a candidate set of programs contains a correct solution with high probability. 
Furthermore, a selective execution strategy is proposed to augment label sets for efficient and effective calibration. 
Extensive experiments across multiple models and benchmarks confirm the method's effectiveness.
While our focus is on code generation,
we believe the proposed method can be adapted to other structured prediction problems.

\noindent\textbf{Limitations.} 
Our method requires executing test cases to augment labels during calibration, which adds computational overhead relative to sample-free baselines.
However, this overhead is confined to the calibration phase and does not impact inference speed on test data.
To further reduce the cost, we introduce a selective execution strategy that allows a bounded label-error rate in exchange for reduced executions. 
The strategy’s effectiveness is supported theoretically (Proposition~\ref{prop:sampling}) and shown empirically (cf. Table~\ref{T:selective verification}).

\nocite{langley00}

\clearpage 

\section*{Broader Impact}
This paper presents work to advance the field of conformal prediction and code generation. Our work has many potential societal consequences, none of which we feel must be specifically highlighted here.

\bibliography{refs}
\bibliographystyle{icml2026}

 
\newpage
\appendix
\onecolumn



\section{Proofs}

\subsection{Proof of Lemma~\ref{lemma:p-value}}
\begin{proof} 
Let $A=\mathrm{Bin}(n,\alpha)$ and $B= n\widehat{R}_n^s(\lambda_k)$. Since the risk in Eq.~\eqref{eq:risk} is an average of binary outcomes, $B$ is also binomial with sample size $n$ and some (unknown) success probability $\alpha'$.
Let $F_A$ and $F_B$ denote the CDFs of $A$ and $B$, respectively. 
Under the null hypothesis $\mathcal{H}_k:R^s(\lambda_k)>\alpha$ on test instances, 
we have $\alpha'>\alpha$.
This implies that $B$ stochastically dominates $A$, i.e., $F_A(t)\geq F_B(t)$ for all $t$.
Let $C=p_k^{BT}=eF_A(B)$, then
\begin{equation*}
\begin{aligned}
\Prob(C\leq c)&=\Prob(F_A(B)\leq \frac{c}{e})\\
&\leq \Prob(F_B(B)\leq \frac{c}{e}).\\
\end{aligned}
\end{equation*}
By the definition of a CDF, we have
\begin{equation*}
\begin{aligned}
\Prob(F_B(B)\leq \frac{c}{e})&= \Prob(B\leq F_B^{-1}(\frac{c}{e}))\\
&= F_B\!\left(F_B^{-1}\!\left(\frac{c}{e}\right)\right)\\
&=\frac{c}{e}\leq c.
\end{aligned}
\end{equation*}
Therefore, under $\mathcal{H}_k$, $p_k^{BT}$ is super-uniform and hence a valid $p$-value.

Furthermore, by Bentkus inequality~\cite{bentkus2004hoeffding}, we obtain
\begin{equation*}
\Prob(\widehat{R}_n^s(\lambda_k)\leq t) \leq e\Prob\big(\mathrm{Bin}(n,\alpha)\leq \lceil nt\rceil\big).
\end{equation*}  
This inequality is nearly tight and $nt=\lceil nt\rceil$, because $n\widehat{R}_n^s(\lambda_k)$ is binomial as discussed above.
Therefore, $p_k^{BT}$ is a near-tight upper bound for the CDF of $\widehat{R}_n^s(\lambda_k)$.
\end{proof}

\subsection{Proof of Theorem~\ref{thm:guarantee}}

\begin{proof}
Since the enhanced calibration set is i.i.d., the set losses $\mathbf{1}\{\mathcal{T}_\lambda(X_i)\cap \tilde{Y}_i=\emptyset\}$ in Eq.~\eqref{eq:risk} are also i.i.d. 
By Lemma~\ref{lemma:p-value}, each $p_k^{BT}$ is a valid $p$-value under the null hypothesis $\mathcal{H}_k:R^s(\lambda_k)>\alpha$. 
Given a set-valued function $\mathcal{T}_{\lambda}$, and a FWER-controlling algorithm at level $\delta$, we can apply Theorem~\ref{thm:LTT} to identify $\Lambda_{\mathrm{valid}}$ such that

\begin{equation*}
\Prob ( \sup_{\lambda\in\Lambda_{\mathrm{valid}}}\{R^s(\lambda)\} \leq \alpha\mid \mathcal{D}_\mathrm{cal}  ) \geq 1-\delta
\end{equation*}
on test programs. Therefore, we have
\begin{equation*}
\Prob ( R^s(\hat{\lambda}) \leq \alpha\mid \mathcal{D}_\mathrm{cal}  ) \geq 1-\delta,
\end{equation*}
as Eq.~\eqref{eq:thm1} holds for any $\lambda\in \Lambda_{\mathrm{valid}}$. 
In particular, it holds for $\hat{\lambda}$ selected by Eq.~\eqref{eq:lambda}.

\end{proof}

\subsection{Proof of Lemma~\ref{lemma:Hbound}} 
\begin{proof}
We first analyze the properties of the random variables $Z_j(u)$ constructed in Algorithm~\ref{alg:EUB}.
The random variables $Z_1(u),\dots,Z_h(u)$ are independent, because index $i_j$ is uniformly sampled from $[M]$, and $\xi_{i_j}$ is independent of other random variables.
Since binary loss $\ell(S_{i_j})$ is bounded by $[0,1]$, if $\xi_{i_j}=1$, we have

\begin{equation*}
Z_j=\frac{\ell(S_{i_j})}{\omega_{i_j}}\leq \frac{1}{\min_j \omega_{i_j}} = \frac{1}{\omega_\mathrm{min}}.
\end{equation*}
If $\xi_{i_j}=0$, $Z_j=0$.
Let $R=1/\omega_\mathrm{min}$. Then $Z_j\in[0,R]$ for all $j\in[h]$.

We denote $\bar{Z}(u)=\frac{1}{h}\sum_{j=1}^hZ_j(u)$, and compute the expected $\bar{Z}(u)$ as

\begin{equation*}
\begin{aligned}
\E[\bar{Z}(u)] &= \frac{1}{h}\sum_{j=1}^h\E[Z_j(u)]\\
&= \frac{1}{h}\sum_{j=1}^h\Big[\omega_{i_j}\cdot\frac{\ell(S_{i_j})}{\omega_{i_j}}\mathbf{1}\{U_{i_j}\leq u\} + (1-\omega_{i_j})\cdot 0\Big]\\
&= \frac{1}{h}\sum_{j=1}^h\ell(S_{i_j})\mathbf{1}\{U_{i_j}\leq u\} = L(u).
\end{aligned}
\end{equation*}

By Hoeffding's inequality and $L(u)=\E[\bar{Z}(u)]$, we obtain

\begin{equation}\label{eq:Hoeffding}
\Prob\big(\lvert \hat{\mu}_Z(u)-L(u) \rvert \geq t \big) \leq 2\exp\Big(-\frac{2ht^2}{R^2} \Big),
\end{equation} 
where $t>0$ is arbitrary.
Setting the right-hand side equal to $\gamma$ yields

\begin{equation*}
2\exp\Big(-\frac{2ht^2}{R^2} \Big) = \gamma\quad \Rightarrow \quad t = R\sqrt{\frac{\log(2/\gamma)}{2h}}.
\end{equation*}
Substituting this value of $t$ into Eq.~\ref{eq:Hoeffding} yields

\begin{equation*}
\Prob\Big(\hat{\mu}_Z(u)-L(u) \leq -R\sqrt{\frac{\log(2/\gamma)}{2h}} \Big) \leq \frac{\gamma}{2}\leq \gamma.
\end{equation*} 

Equivalently, 

\begin{equation*}
\Prob\Big(L(u) \leq \hat{\mu}_Z(u) + R\sqrt{\frac{\log(2/\gamma)}{2h}} \Big) \geq 1 - \gamma.
\end{equation*} 

Let $\hat{L}^u(\gamma)=\hat{\mu}_Z(u) + R\sqrt{\log(2/\gamma)/2h}$. Then
\begin{equation*}
\Prob(L(u)\leq \hat{L}^u(\gamma)) \geq 1-\gamma,
\end{equation*} 
which completes the proof. 
\end{proof}

\subsection{Proof of Proposition~\ref{prop:sampling}} \label{A:prop1}

\begin{proof}
We prove Proposition~\ref{prop:sampling}.  
In Section~\ref{sec:sv}, we select $\hat{u} = \max\{u\in[0,1]:\hat{L}^u(\gamma)\leq\epsilon\}$, and by Lemma~\ref{lemma:Hbound} we have the following upper bound:

\begin{equation}\label{eq:upper bound}
L(u)\leq \hat{L}^u(\gamma)=\epsilon.
\end{equation} 

According to Theorem~\ref{thm:guarantee}, for a test instance $(X_{n+1},Y_{n+1})$, we have
\begin{equation*}
\Prob \big(\Prob(\mathcal{T}_{\hat{\lambda}}(X_{n+1})\cap Y_{n+1}=\emptyset)\leq\alpha)  \big) \geq 1-\delta.
\end{equation*}

In this guarantee, let $Y_{n+1}=Y_{n+1}^t\cup Y_{n+1}^f$, where $Y_{n+1}^t$ and $Y_{n+1}^f$ are disjoint sets and $\lvert Y_{n+1}^t \lvert/\lvert Y_{n+1} \lvert\geq 1-\epsilon$ according to Eq.~\eqref{eq:upper bound}. 
For any element $y\in Y_{n+1}$, we assume $y\in \mathcal{T}_{\hat{\lambda}}(X_{n+1})$ independently with equal probability $p$. Then,
\begin{equation*}
\Prob(\mathcal{T}_{\hat{\lambda}}(X_{n+1})\cap Y_{n+1}=\emptyset)=(1-p)^{\lvert Y_{n+1} \lvert} \leq \alpha,
\end{equation*} 

Consequently,
\begin{equation*}
\begin{aligned}
\Prob(\mathcal{T}_{\hat{\lambda}}(X_{n+1})\cap Y_{n+1}^t=\emptyset) &= (1-p)^{\lvert Y_{n+1}^t \lvert}=\Big((1-p)^{\lvert Y_{n+1} \lvert}\Big)^{\frac{\lvert Y_{n+1}^t \lvert}{\lvert Y_{n+1} \lvert}}\\
&\leq  \alpha^{1-\epsilon}
\leq 1-(1-\alpha)(1-\epsilon)\\
&=\alpha+\epsilon(1-\alpha)
\end{aligned}
\end{equation*} 
holds with probability at least $(1-\gamma)\cdot(1-\delta)$ (since the bounds involving $\epsilon$ and $\alpha$ hold with probabilities at least $1-\gamma$ and $1-\delta$, respectively). 
Additionally, we use the Bernoulli inequality for $\alpha,\epsilon\in(0,1)$ in the last inequality.
\end{proof}

\subsection{Proof of Proposition~\ref{prop:CLT}}
\label{A:proof for CLT}
\begin{proof}
Let $h$ be the number of sample programs from LLMs.
By the Lindeberg–Feller central limit theorem~\cite{billingsley1961lindeberg}, we obtain
\begin{equation*}
\frac{\sqrt{h}(L(u)-\hat{\mu}_Z(u))}{\sigma_Z(u)}\rightarrow\mathcal{N}(0,1)
\end{equation*} 
as $h\rightarrow \infty$, since the $Z_j(u)$ are i.i.d.
By the weak law of large numbers, we have $\hat{\sigma}_Z(u) \xrightarrow{p} \sigma_Z(u),$
and thus,
an application of Slutsky's theorem gives
\begin{equation*}
\Prob\Big(\frac{L(u)-\hat{\mu}_Z(u)}{\sqrt{\hat{\sigma}^2_Z(u)/h}} \leq z_{1-\gamma}\Big)\rightarrow 1-\gamma,
\end{equation*} 
because $z_{1-\gamma}$ is the $(1-\gamma)$-quantile of the standard normal distribution.
Equivalently, when $h$ approaches positive infinity, we have
\begin{equation*}
\Prob\Big(L(u) \leq \hat{\mu}_Z(u) + z_{1-\gamma}\sqrt{\hat{\sigma}^2_Z(u)/h}\Big)\rightarrow 1-\gamma,
\end{equation*} 
which completes the proof. 
\end{proof}

\section{FWER-controlling algorithm}\label{A:FWER}

After obtaining valid p-values by Eq.~\eqref{eq:p-value}, we discuss procedures for computing $\Lambda_{\mathrm{valid}}=\mathcal{A}(p_1^{BT},\dots,p_N^{BT})$ on structured risks. 
Since evaluating multiple hypotheses simultaneously increases the chance of false positives, we require family-wise error rate (FWER) correction.

\begin{definition}
(FWER control) Let $\mathcal{A}(p_1^{BT},\dots,p_N^{BT})$ be an algorithm that outputs a subset $\Lambda_{\mathrm{valid}}\subseteq\Lambda$. We say $\mathcal{A}$ controls the FWER at $(\alpha,\delta)$-level if
\begin{equation*}
\Prob \Big( \sup_{\lambda \in \Lambda_{\mathrm{valid}}}\Prob(R^s(\lambda)\leq\alpha)  \Big) \geq 1-\delta,
\end{equation*}
where the supremum over $\emptyset$ is defined as $-\infty$.
\end{definition}

We provide three alternative FWER-controlling algorithms: the \textit{Bonferroni method}, the \textit{Holm--Bonferroni method}~\cite{holm1979simple}, and \textit{fixed sequence testing}~\cite{bauer1991multiple}.
The Bonferroni method defines $\Lambda_{\mathrm{valid}}:=\{\lambda_k\mid p_k^{BT}\leq \delta/N\}_{k=1}^N$.
The Holm--Bonferroni method is an improved version of the Bonferroni method.
It sorts p-values as $p_{(1)}\leq\dots\leq p_{(N)}$ and $\Lambda_{\mathrm{valid}}:=\{\lambda_{(1)},\dots,\lambda_{(k)}\}$, where $k\in[N]$ is the minimum index satisfying $p_{(k)}\leq\delta/(N-k+1)$ and $p_{(k+1)}>\delta/(N-k)$.

For fixed sequence testing, we leverage the prior knowledge that partial programs with more nodes removed correspond to larger prediction sets. 
This implies that smaller values of $\lambda$ are more likely to satisfy the target risk. 
In the extreme case where $\lambda=0$, the partial program prunes the entire tree, trivially achieving a risk of zero in Eq.~\eqref{eq:risk}. 
Therefore, a heuristic correction strategy is to test $\lambda$ in increasing order and stop at the first accepted hypothesis. 
To broaden the search, this process can be repeated from different predefined starting points (e.g., a coarse uniform grid in $\Lambda$). 
Let the parameter set $\Lambda=\{\lambda_1,\dots,\lambda_N\}$ be in descending order (i.e., $\lambda_1\geq\dots\geq\lambda_N$). The full fixed sequence testing procedure for code is summarized in Algorithm~\ref{alg:fst}.



\begin{algorithm} 
\caption{Fixed sequence testing} 
\label{alg:fst}
\begin{algorithmic}[1]
\STATE {\bfseries Input:} parameter set $\Lambda=\{\lambda_1,\dots,\lambda_N\}$; error level $\delta\in(0,1)$; p-values $\{p_k^{BT}\}_{k=1}^N$; predefined starting points $\mathcal{I}\subset[N]$;
\STATE {\bfseries Output:} reject set $\Lambda_{\mathrm{valid}}$.
\STATE Initialize a reject set $\Lambda_{\mathrm{valid}}=\emptyset$;
\FOR {$k\in \mathcal{I}$}
\IF[\emph{Avoid duplicates}]{$\lambda_k\notin\Lambda_{\mathrm{valid}}$} 
\WHILE{$p_k^{BT}\leq \delta/\lvert\mathcal{I} \rvert$}
\STATE $\Lambda_{\mathrm{valid}}=\Lambda_{\mathrm{valid}}\cup\{\lambda_k\}$; 
\STATE $k=k-1$;
\ENDWHILE
\ENDIF
\ENDFOR
\end{algorithmic}
\end{algorithm}



\section{Asymptotic upper bound based on CLT}\label{A:CLT bound}

In addition to the non-asymptotic error upper bound used in Lemma~\ref{lemma:Hbound}, we also consider an asymptotic bound for settings where the number of sampled programs $M$ is large.

In Section~\ref{sec:sv}, we estimate an unbiased upper bound for $L(u)$ by constructing a confidence interval for the average of $\{Z_j(u)\}_{j=1}^h$.
Based on the central limit theorem (CLT), we obtain the following asymptotic upper bound.
\begin{proposition}[Asymptotic upper bound]\label{prop:CLT}
Given a threshold $u$ and the sample size $h$, let $Z_j(u)$ be the i.i.d. random variables defined in Algorithm~\ref{alg:CLT} for $j\in [h]$.
Define the error upper bound based on CLT as

\begin{equation*}
\hat{L}^u(\gamma):=\hat{\mu}_Z(u) + z_{1-\gamma}\cdot\frac{\hat{\sigma}_Z(u)}{\sqrt{h}},
\end{equation*}
where $\hat{\mu}$ and $\hat{\sigma}$ denote the empirical mean and standard deviation of $\{Z_j(u)\}_{j=1}^h$, respectively, and $z_{1-\gamma}$ is the $(1-\gamma)$-quantile of the standard normal distribution.
Then, we have

\begin{equation*}
\mathop{\lim\inf} \limits_{m\rightarrow\infty}\Prob(L(u)\leq \hat{L}^u(\gamma)) \geq 1-\gamma.
\end{equation*} 
\end{proposition}

The proof is provided in Appendix~\ref{A:proof for CLT}, and the overall procedure is described in Algorithm~\ref{alg:CLT}.

\begin{algorithm} 
\caption{Selective Verification with Asymptotic Error Bound} 
\label{alg:CLT}
\begin{algorithmic}[1]
\STATE {\bfseries Input:} unlabeled sampling programs $\{X_i\}_{i=1}^M$; uncertainty scores $\{U_i\}_{i=1}^M$; test cases $\{C_i\}_{i=1}^M$; sampling weights $\{\omega_i\}_{i=1}^M$; sampling size $h$; confidence level $\gamma$;
\STATE {\bfseries Output:} error upper bound $\hat{L}^u(\gamma)$.
\FOR {$j=1,\dots,h$}
\STATE Sample an index $i_j\sim \mathrm{Unif}([M])$;
\STATE Sample a random variable $\xi_{i_j}\sim \mathrm{Bern}(\omega_{i_j})$;
\IF{$\xi_{i_j}=1$}
\STATE Run the test case $C_{i_j}$ to get $S_{i_j}$;
\STATE Compute the importance-weighted error using $Z_j=\ell(S_{i_j})/\omega_{i_j}$;
\ELSE
\STATE $Z_j=0$;
\ENDIF
\ENDFOR
\STATE Given a threshold $u$, define $Z_j(u)=Z_j\cdot\mathbf{1}\{U_{i_j}\leq u\}$ for $j\in[h]$;
\STATE $\hat{\mu}_Z(u)\leftarrow \frac{1}{h}\sum_{j=1}^{h}Z_j(u)$;
\STATE $\hat{\sigma}_Z(u)\leftarrow \sqrt{\frac{1}{h}\sum_{j=1}^{h}(Z_j(u)-\hat{\mu}_Z(u))^2}$;
\STATE $z_{1-\gamma}\leftarrow$ $(1-\gamma)$-quantile of the standard normal distribution;
\STATE $\hat{L}^u(\gamma)\leftarrow \hat{\mu}_Z(u) + z_{1-\gamma}\cdot\frac{\hat{\sigma}_Z(u)}{\sqrt{h}}$.
\end{algorithmic}
\end{algorithm}

\begin{figure*}[t]
 \centering
 \subfigure[Coverage distribution]{\label{F:coverage}
    \includegraphics[width=2.1in]{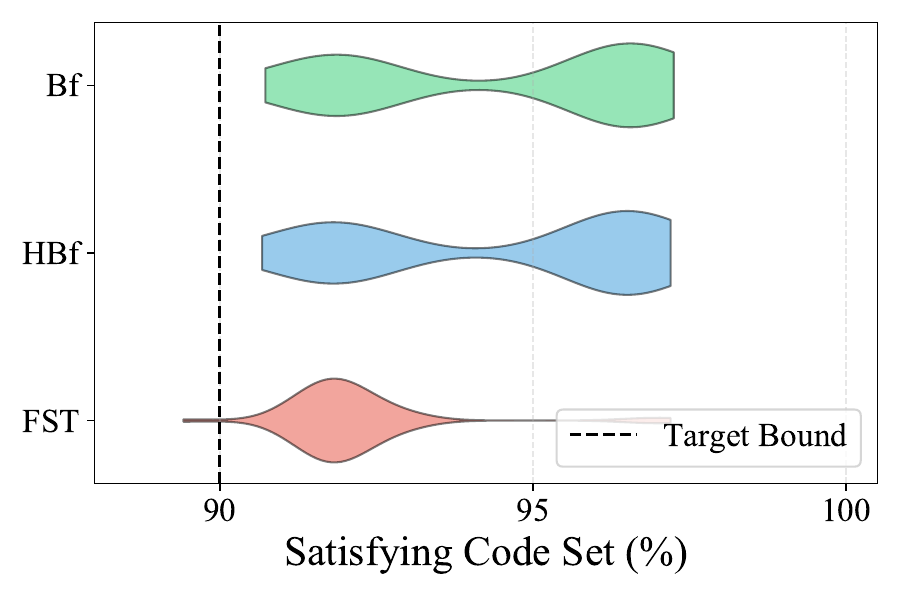}}
 \subfigure[Sensitivity analysis on $t_{\max}$]{\label{F:t-remove}
     \includegraphics[width=2.1in]{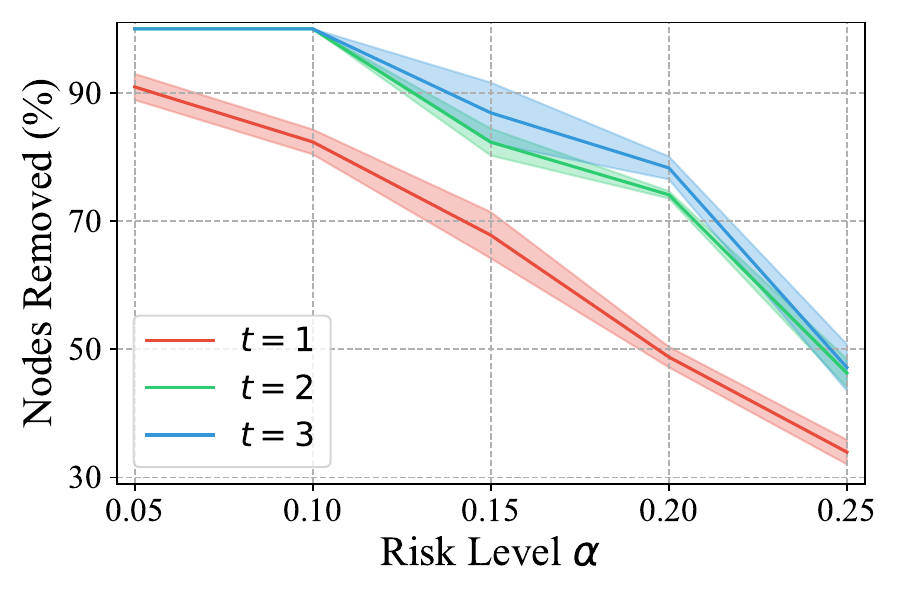}}
 \subfigure[Sensitivity analysis on $t_{\max}$]{\label{F:t-set}
    \includegraphics[width=2.1in]{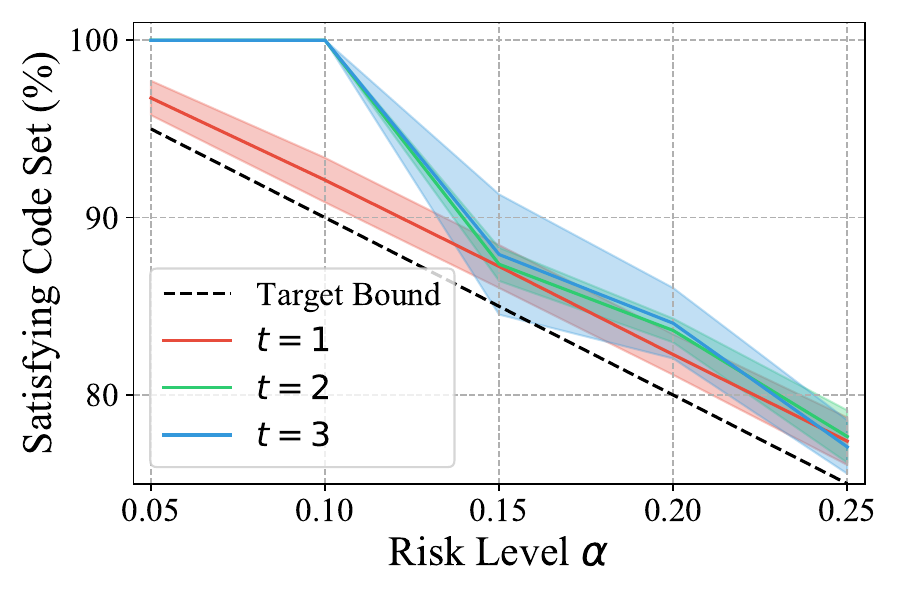}}
 \caption{Fig.~\ref{F:coverage} is the distribution of code set coverage over 100 trials for three different FWER algorithms, respectively; Fig.~\ref{F:t-remove} and~\ref{F:t-set} are parameter sensitivity results of maximum subtree removals $t_{\max}$ w.r.t risk level $\alpha$.
 These experiments are conducted on MBPP.}
\end{figure*}

\begin{table*}[t]
    \centering
    \caption{Performance of our method with different FWER algorithms when $\alpha=0.1$. The results are the average over 100 random splits. The smaller node removal is better, when the code set coverage exceeds the 0.9.}
      \begin{sc}
      \begin{tabular}{lccc}
      \toprule
       {FWER alg.} & {Bf} & {HBf} & {FST}  \\
      \midrule 
      {Removal} & $87.9{\scriptstyle\pm 3.1}$ & $87.0{\scriptstyle\pm 3.3}$ & $82.3{\scriptstyle\pm 1.9}$\\
      {Coverage} & $94.7{\scriptstyle\pm 2.2}$ & $94.4{\scriptstyle\pm 2.3}$ & $92.1{\scriptstyle\pm 1.3}$\\
      \bottomrule
      \end{tabular}
      \end{sc}
    \label{T:FWER}
\end{table*}

\section{Additional Results}
\label{A:exp}

\subsection{Different FWER-controlling algorithm}
To evaluate the performance of different FWER algorithms on our method, we plot the distribution of code set coverage ($\alpha=0.1$) as a violin plot over 100 random splits of MBPP in Fig.~\ref{F:coverage}, and list the mean results in Table.~\ref{T:FWER}.
The results demonstrate that compared to fixed sequence testing (FST), Bonferroni (Bf) and Holm–Bonferroni (HBf) methods are more conservative and remove more nodes, as confirmed by the results in Table~\ref{T:FWER}.
This indicates the importance of program prior knowledge used in FST.

\subsection{Parameter sensitivity analysis on $t_{\max}$}
We next perform parameter sensitivity analysis on the maximum subtree removals $t_{\max}$. 
The results on MBPP are shown in Fig.~\ref{F:t-remove} and Fig.~\ref{F:t-set}.
We observe that the best result is achieved when $t_{\max}=1$, and $t_{\max}\geq2$ leads to more node removals than $t_{\max}=1$.
A plausible explanation is that when models are not well-calibrated, a larger $t_{\max}$ 
distributes pruned nodes and may lead to localized errors that cannot be fully corrected.

\subsection{Sample outputs}

\begin{figure*}[t]
 \centering
 \subfigure{
 \includegraphics[width=3.7in]{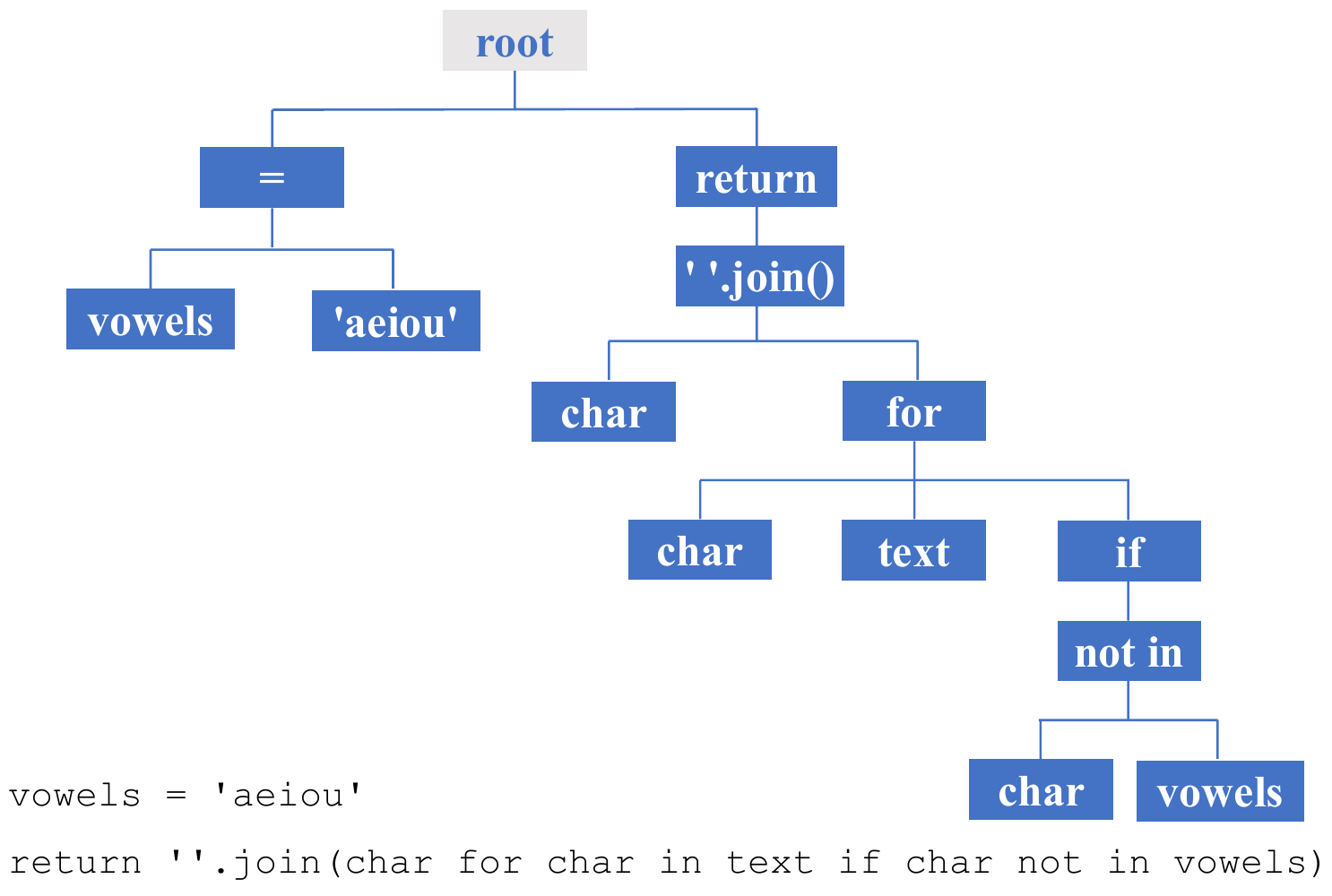}}
 \subfigure{
\includegraphics[width=3.7in]{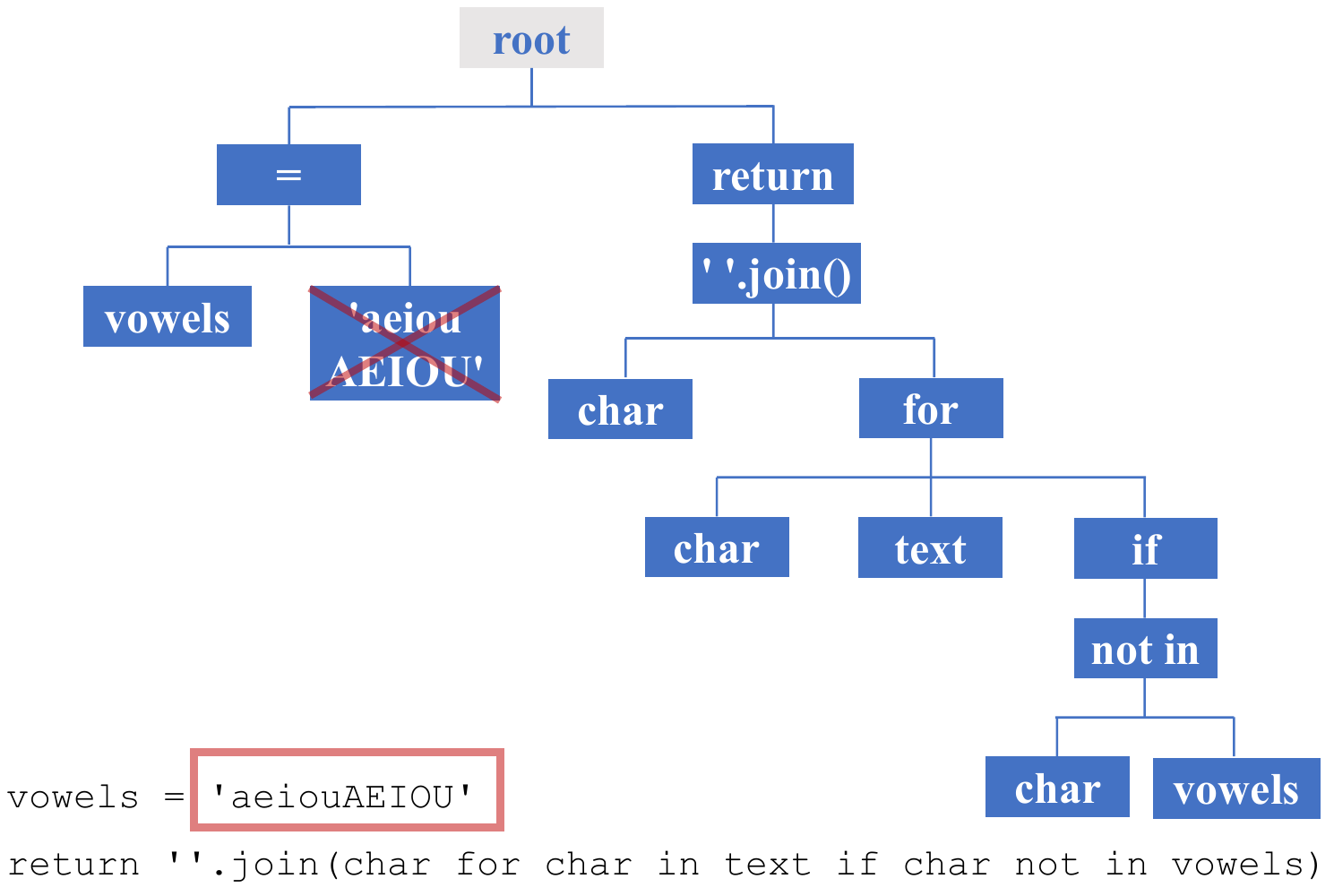}}
 \caption{An illustrative example from HumanEval. The top part is a correct code snippet, and the bottom part is a generated but incorrect one.
 \name removes one node in the AST, resulting a prediction set (i.e., a partial program) that contains  the correct program.}

\end{figure*}

\begin{figure*}[t]
 \centering
 \subfigure{
 \includegraphics[width=2.2in]{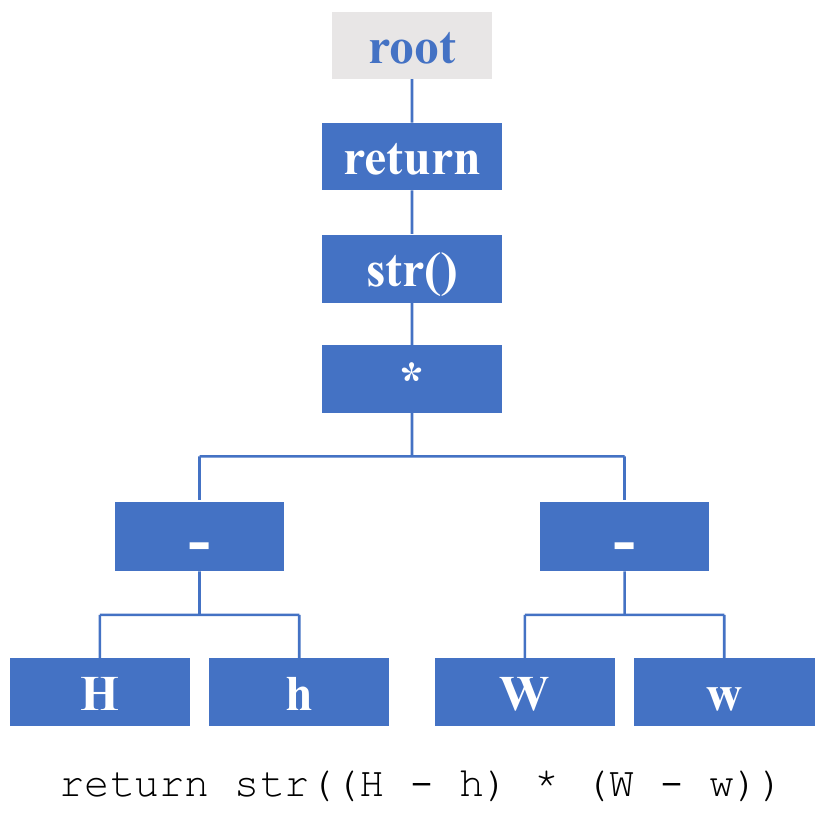}}
 \hspace{0.5in}
 \subfigure{
\includegraphics[width=2.2in]{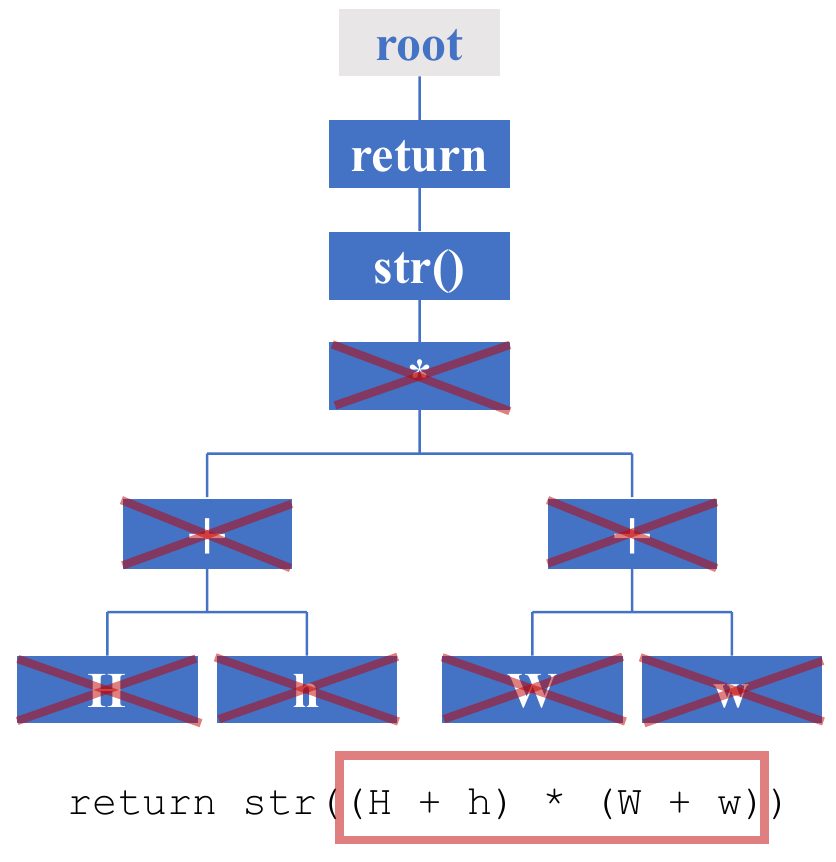}}
 \caption{An illustrative example from APPS. The left part is a correct code snippet, and the right part is a generated but incorrect one.
 \name removes seven nodes in the AST, resulting a prediction set (i.e., a partial program) that contains  the correct program.}

\end{figure*}

\begin{figure*}[t]
 \centering
 \subfigure{
 \includegraphics[width=4.0in]{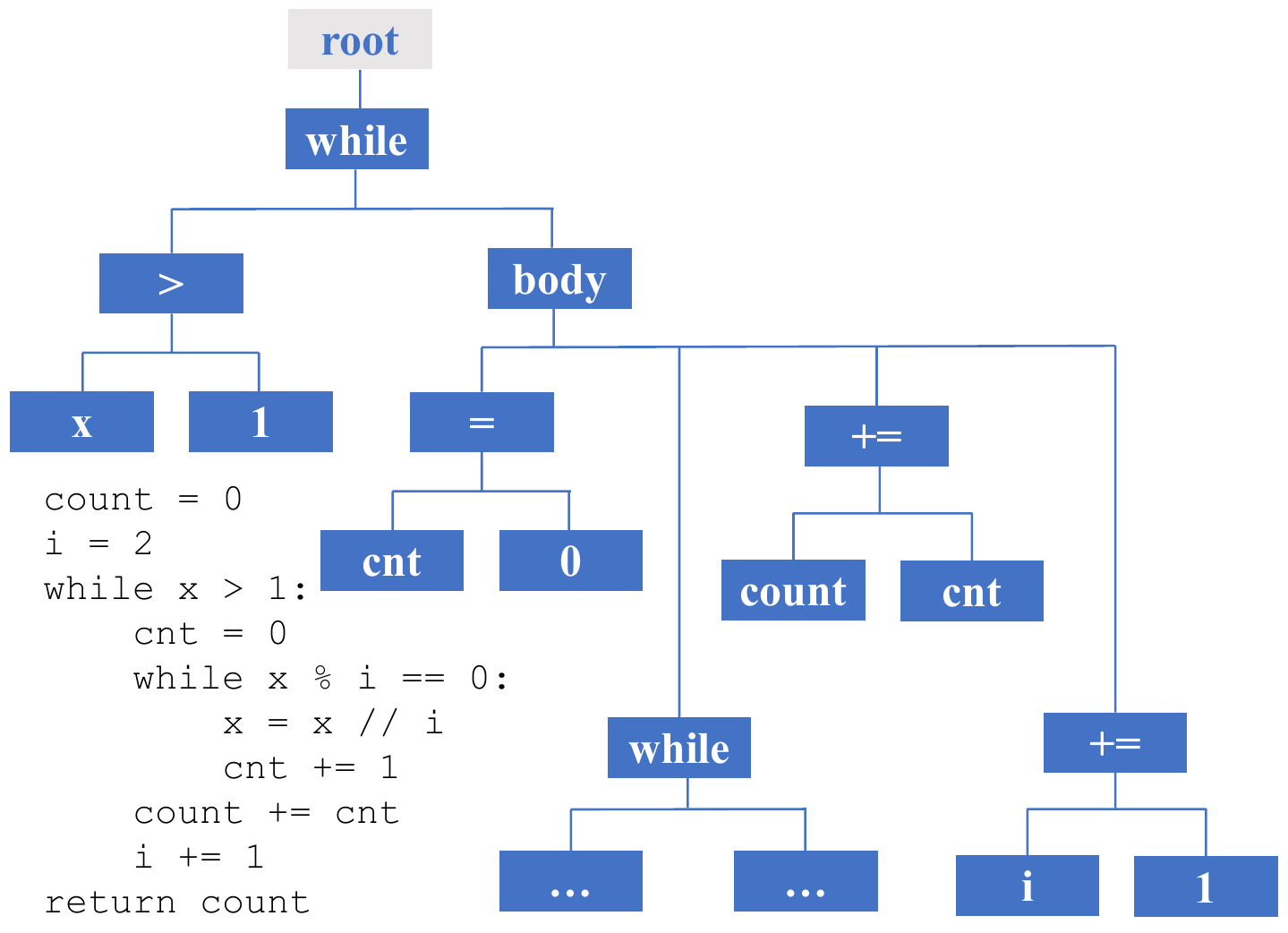}}
 \subfigure{
\includegraphics[width=4.4in]{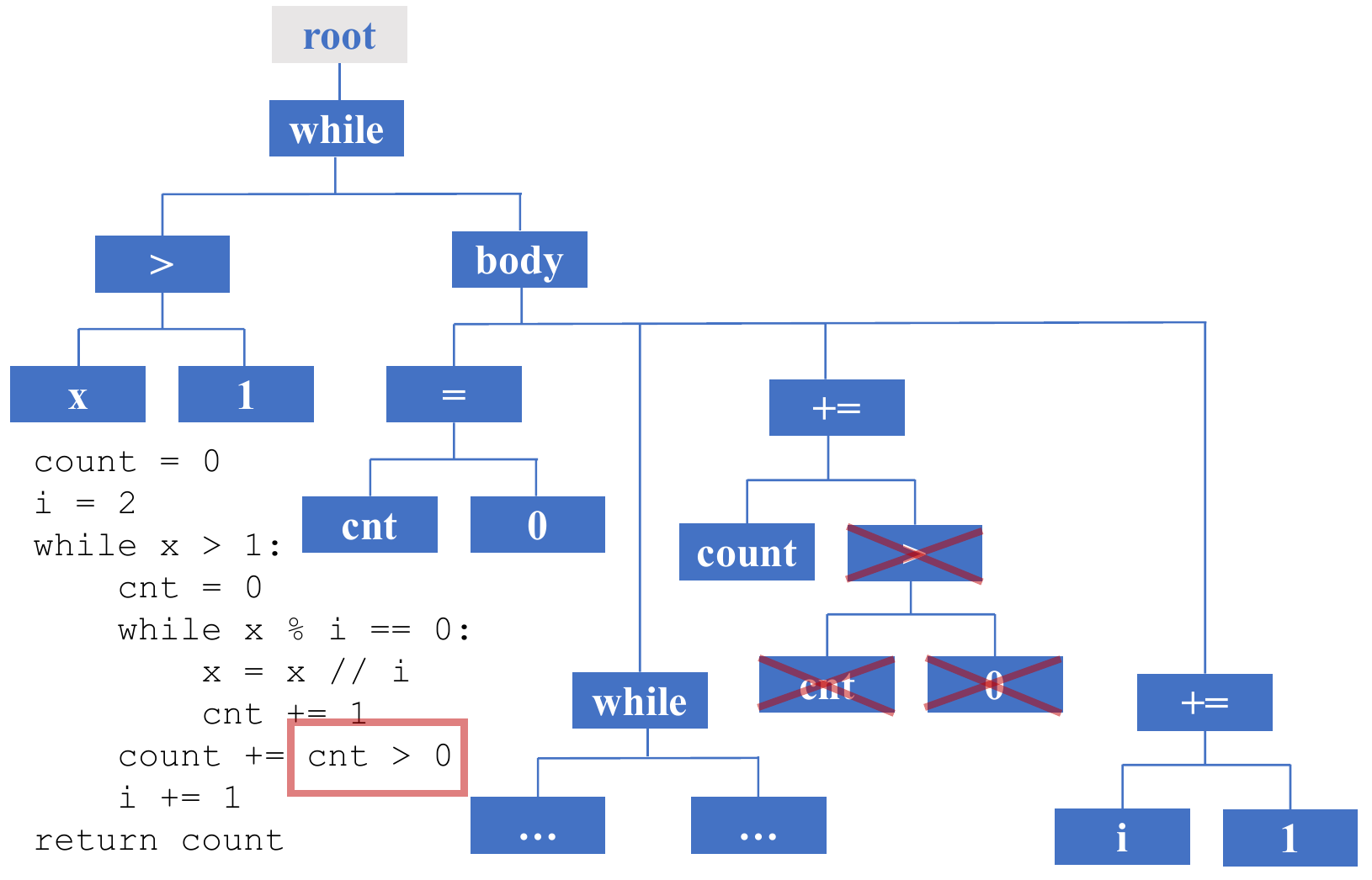}}
 \caption{An illustrative example from APPS. The top part is a correct code snippet, and the bottem part is a generated but incorrect one.
 \name removes three nodes in the AST (we omit the overlapping AST structures for brevity), resulting a prediction set (i.e., a partial program) that contains  the correct program.}

\end{figure*}

\end{document}